\PassOptionsToPackage{unicode}{hyperref}
\PassOptionsToPackage{hyphens}{url}
\PassOptionsToPackage{dvipsnames,svgnames,x11names}{xcolor}
\documentclass[12pt]{article}

\usepackage{algorithm}
\usepackage{bbm}
\usepackage{algorithmic}
\usepackage{tikz}
\usetikzlibrary{matrix,positioning}
\usepackage{booktabs}
\usepackage{multirow}
\usepackage{placeins}
\usepackage{array}
\usepackage{siunitx}
\usepackage{amsmath,amssymb}
\usepackage[percent]{overpic}
\usepackage{amsthm}
\usepackage{mathrsfs}
\usepackage{iftex}
\ifPDFTeX
  \usepackage[T1]{fontenc}
  \usepackage[utf8]{inputenc}
  \usepackage{textcomp} 
\else 
  \usepackage{unicode-math}
  \defaultfontfeatures{Scale=MatchLowercase}
  \defaultfontfeatures[\rmfamily]{Ligatures=TeX,Scale=1}
\fi
\usepackage{lmodern}
\ifPDFTeX\else  
\fi
\IfFileExists{upquote.sty}{\usepackage{upquote}}{}
\IfFileExists{microtype.sty}{
  \usepackage[]{microtype}
  \UseMicrotypeSet[protrusion]{basicmath} 
}{}
\makeatletter
\@ifundefined{KOMAClassName}{
  \IfFileExists{parskip.sty}{%
    \usepackage{parskip}
  }{
    \setlength{\parindent}{0pt}
    \setlength{\parskip}{6pt plus 2pt minus 1pt}}
}{
  \KOMAoptions{parskip=half}}
\makeatother
\usepackage{xcolor}
\makeatletter
\ifx\paragraph\undefined\else
  \let\oldparagraph\paragraph
  \renewcommand{\paragraph}{
    \@ifstar
      \xxxParagraphStar
      \xxxParagraphNoStar
  }
  \newcommand{\xxxParagraphStar}[1]{\oldparagraph*{#1}\mbox{}}
  \newcommand{\xxxParagraphNoStar}[1]{\oldparagraph{#1}\mbox{}}
\fi
\ifx\subparagraph\undefined\else
  \let\oldsubparagraph\subparagraph
  \renewcommand{\subparagraph}{
    \@ifstar
      \xxxSubParagraphStar
      \xxxSubParagraphNoStar
  }
  \newcommand{\xxxSubParagraphStar}[1]{\oldsubparagraph*{#1}\mbox{}}
  \newcommand{\xxxSubParagraphNoStar}[1]{\oldsubparagraph{#1}\mbox{}}
\fi
\makeatother

\usepackage{longtable,booktabs,array}
\usepackage{calc} 
\usepackage{etoolbox}
\makeatletter
\patchcmd\longtable{\par}{\if@noskipsec\mbox{}\fi\par}{}{}
\makeatother
\IfFileExists{footnotehyper.sty}{\usepackage{footnotehyper}}{\usepackage{footnote}}
\makesavenoteenv{longtable}
\usepackage{graphicx}
\makeatletter
\def\maxwidth{\ifdim\Gin@nat@width>\linewidth\linewidth\else\Gin@nat@width\fi}
\def\maxheight{\ifdim\Gin@nat@height>\textheight\textheight\else\Gin@nat@height\fi}
\makeatother
\setkeys{Gin}{width=\maxwidth,height=\maxheight,keepaspectratio}
\makeatletter
\def\fps@figure{htbp}
\makeatother

\makeatletter
\@ifpackageloaded{caption}{}{\usepackage{caption}}
\AtBeginDocument{%
\ifdefined\contentsname
  \renewcommand*\contentsname{Table of contents}
\else
  \newcommand\contentsname{Table of contents}
\fi
\ifdefined\listfigurename
  \renewcommand*\listfigurename{List of Figures}
\else
  \newcommand\listfigurename{List of Figures}
\fi
\ifdefined\listtablename
  \renewcommand*\listtablename{List of Tables}
\else
  \newcommand\listtablename{List of Tables}
\fi
\ifdefined\figurename
  \renewcommand*\figurename{Figure}
\else
  \newcommand\figurename{Figure}
\fi
\ifdefined\tablename
  \renewcommand*\tablename{Table}
\else
  \newcommand\tablename{Table}
\fi
}
\@ifpackageloaded{float}{}{\usepackage{float}}
\floatstyle{ruled}
\@ifundefined{c@chapter}{\newfloat{codelisting}{h}{lop}}{\newfloat{codelisting}{h}{lop}[chapter]}
\floatname{codelisting}{Listing}

\makeatother
\makeatletter
\@ifpackageloaded{caption}{}{\usepackage{caption}}
\@ifpackageloaded{subcaption}{}{\usepackage{subcaption}}
\makeatother

\ifLuaTeX
  \usepackage{selnolig}  
\fi
\usepackage[]{natbib}
\usepackage{bookmark}

\IfFileExists{xurl.sty}{\usepackage{xurl}}{} 
\hypersetup{
  pdftitle={Robust Bayesian Inference for Unnormalized Statistical Models},
  pdfauthor={Jiongran Wang; Debdeep Pati; Anirban Bhattacharya},
  pdfkeywords={doubly-intractable models, exponentially tilted empirical likelihood, model misspecification, semiparametric Bayes, score matching},
  colorlinks=true,
  linkcolor={blue},
  filecolor={Maroon},
  citecolor={Blue},
  urlcolor={Blue},
  pdfcreator={LaTeX via pandoc}}

\DeclareMathOperator*{\argmin}{arg\,min}

\def \convd{\overset{d} \to}
\def \convp{\overset{p} \to}

\newcommand{\be}{\begin{equs}}
\newcommand{\ee}{\end{equs}}

\numberwithin{equation}{section}
\theoremstyle{plain}
\newtheorem{theorem}{Theorem}
\newtheorem{assumption}{Assumption}
\newtheorem{remark}{Remark}

\newtheorem{lemma}{Lemma}
\newtheorem{proposition}{Proposition}

\newcommand{\anon}{1}

\begin{document}

\def\spacingset#1{\renewcommand{\baselinestretch}%
{#1}\small\normalsize} \spacingset{1}


\if1\anon
{
  \title{\bf Robust Bayesian Inference for Unnormalized Models with Mixed-Domain Data}
  \author{Jiongran Wang\\
    Department of Statistics, Texas A\&M University\\
    and \\
    Debdeep Pati \\
    Department of Statistics, University of Wisconsin-Madison\\
    and \\
    Anirban Bhattacharya \\
    Department of Statistics, Texas A\&M University}
  \maketitle
} \fi

\if0\anon
{
  \bigskip
  \bigskip
  \bigskip
  \begin{center}
    {\LARGE\bf Robust Bayesian Inference for Unnormalized Models with Mixed-Domain Data}
\end{center}
  \medskip
} \fi

\bigskip
\begin{abstract}
Many statistical models involve parameter-dependent normalizing constants that are computationally intractable, creating substantial obstacles to standard Bayesian inference. Although existing likelihood-based algorithms can often circumvent these constants, their uncertainty quantification may be poorly calibrated under model misspecification. To address these challenges, we propose SME-BETEL, a semiparametric Bayesian framework that combines score matching estimating equations with Bayesian exponentially tilted empirical likelihood. The resulting posterior avoids evaluation of normalizing constants and does not require learning-rate calibration. Building on this framework, we develop a new score matching criterion for mixed-domain data, extending SME-BETEL to models whose observations combine components from different sample spaces. This construction enables robust Bayesian inference for mixed-domain doubly-intractable models. We establish consistency and asymptotic normality of the score matching estimator, and prove a Bernstein–von Mises theorem for the SME-BETEL posterior. These results show that SME-BETEL credible sets are asymptotically calibrated to the sampling variability of the score matching estimator, yielding valid frequentist coverage under model misspecification. Simulation studies show that SME-BETEL remains competitive under correct specification and substantially improves uncertainty quantification under misspecification. An ozone-monitoring application demonstrates the practical utility of the mixed-domain construction for spatial preferential modeling.
\end{abstract}

\noindent%
\bigskip
\spacingset{1.15}

\section{Introduction}\label{sec-intro}

Statistical models of the form $p(\mathbf{x}; \boldsymbol{\theta}) = q(\mathbf{x}; \boldsymbol{\theta}) / Z(\boldsymbol{\theta})$ with an intractable parameter-dependent normalizing constant $Z(\boldsymbol{\theta})$ are prevalent in a multitude of fields including spatial statistics, network analysis, and graphical models \citep{gutmann2013estimation, park2018bayesian}. Bayesian inference on the parameter $\boldsymbol{\theta}$ in such settings is called {\it doubly-intractable} since both the likelihood function and the posterior distribution involve intractable normalizing constants. Traditional MCMC algorithms, such as Metropolis–Hastings (MH), are rendered inapplicable since computing the posterior odds requires evaluating ratios of the form $Z(\boldsymbol{\theta}_1)/Z(\boldsymbol{\theta}_2)$. A variety of methods (see \citet{park2018bayesian} for a review) have been developed for Bayesian inference in doubly-intractable models, including MCMC algorithms that introduce auxiliary variables to cancel this ratio \citep{moller2006efficient,liang2010double,murray2012mcmc,liang2016adaptive,rao2016data,lee2026delayed}, likelihood approximation approaches such as pseudo-marginal MCMC \citep{andrieu2009pseudo,lyne2015russian,yang2025correlated}, and, more recently, contrastive Bayesian inference based on noise contrastive estimation \citep{sonobe2026contrastive}. While these approaches address the computational intractability, Bayesian uncertainty quantification based on a fully specified working model can be poorly calibrated under model misspecification \citep{kleijn2012bernstein}. This motivates robust Bayesian methods that circumvent the intractability of $Z(\boldsymbol{\theta})$ while simultaneously delivering calibrated uncertainty under misspecification. We introduce such an approach in this article.

In an influential article, \citet{hyvarinen2005estimation} proposed {\it score matching} for parameter estimation in unnormalized models without evaluating $Z(\boldsymbol{\theta})$. Beyond its extensive use in score-based generative models \citep{song2019generative,lai2025principles,ranganath2026unified}, score matching has also been connected to Bayesian inference through Bayesian model selection and generalized Bayesian inference \citep{shao2019bayesian, jewson2022general,altamirano2023robust,matsubara2024generalized,bharti2026amortised}. Existing generalized Bayesian procedures based on score-type losses typically require tuning a learning rate, and this choice can strongly affect uncertainty quantification. This motivates an alternative use of score matching that preserves its ability to avoid normalizing constants while avoiding learning-rate calibration.

A calibration-free route to robust inference is provided by exponentially tilted empirical likelihood (ETEL; \citealp{schennach2005bayesian}), a variant of empirical likelihood (EL; \citealp{owen1988empirical,qin1994empirical,owen2001empirical}) designed for inference from moment conditions without requiring full model specification. Both EL and ETEL admit Bayesian formulations \citep{lazar2003bayesian,schennach2005bayesian}, and frequentist properties of Bayesian ETEL have been established by \citet{chib2018bayesian}. Recent work further shows that Bayesian ETEL can yield calibrated semiparametric posteriors from moment conditions, including under misspecification and partial model specification \citep{chib2018bayesian,tang2022bayesian}. Motivated by these properties, we incorporate the first-order optimality conditions of the score matching objective into Bayesian ETEL as moment conditions, a combination that has not been previously exploited for doubly-intractable models to the best of our knowledge. We call the resulting framework SME-BETEL. The corresponding posterior does not require evaluating intractable normalizing constants or calibrating a learning rate, and it provides reliable uncertainty quantification under model misspecification. We verify the latter property theoretically by proving a Bernstein-von Mises (BvM) theorem showing that the induced posterior covariance matches the sampling variance of the score matching estimator. Consequently, SME-BETEL credible sets are asymptotically calibrated to the sampling distribution of the score matching estimator under model misspecification.

Building on the basic SME-BETEL construction, the main methodological contribution of this article is to extend the framework to mixed-domain data. The original score matching objective applies to continuous data on Euclidean domains, while existing extensions accommodate particular nonstandard supports, including nonnegative, compact, and manifold spaces \citep{hyvarinen2007some,mardia2016score,yu2019generalized,liu2022estimating,scealy2023score}. These methods, however, are primarily designed for observations on a single type of support and do not directly address settings in which different components of the same observation lie on different domains. Such mixed-domain structures arise naturally in preferential sampling \citep{diggle2010geostatistical,pati2011bayesian} and marked spatial point processes \citep{eckardt2023marked}, where Euclidean-valued responses or marks are observed together with locations on a bounded spatial domain and the sampling density may involve an intractable normalizing integral. Existing score matching constructions therefore do not directly provide the moment conditions needed for SME-BETEL in such settings. To address this gap, we derive a new mixed-domain score matching objective that combines the usual Euclidean score matching construction with a boundary-adapted treatment of the bounded-domain component. Its first-order conditions can then be incorporated directly into SME-BETEL, extending robust Bayesian inference to mixed-domain doubly-intractable models without evaluating the spatial normalizing integral. We demonstrate this construction through a misspecified preferential sampling simulation and an application to eastern United States ozone monitoring data.


The remainder of the paper is organized as follows. Section~\ref{SME-BETEL} reviews background material, introduces SME-BETEL, and develops its extension to mixed-domain data. Section~\ref{sec:simulation_study} presents a misspecified preferential sampling simulation, while Section~\ref{sec:sme_betel} establishes the theoretical properties of SME-BETEL. Section~\ref{real_data} applies SME-BETEL to eastern United States ozone data. The Supplemental Material contains proofs of the theoretical results and extensive simulation studies under correct specification and misspecification, including models with Euclidean, compact, spherical, and nonnegative support.

\section{SME-BETEL: Bayesian ETEL Based on Score Matching First-Order Conditions}
\label{SME-BETEL}

In this section, we introduce the SME-BETEL framework. We first present the problem setup and review score matching \citep{hyvarinen2005estimation}, which provides a key building block for SME-BETEL. We then develop the SME-BETEL methodology and extend its construction to accommodate mixed-domain data.

\paragraph*{Background} Suppose we observe $n$ independent observations $\mathbf{x}_1, \ldots, \mathbf{x}_n \in \mathbb{R}^m$ from a distribution with density $p_{\mathbf{x}}(\cdot)$. To model the individual observations, we introduce a parametric working model $p(\mathbf{x};\boldsymbol{\theta}) = q(\mathbf{x};\boldsymbol{\theta}) / Z(\boldsymbol{\theta})$, with $\mathbf{x} \in \mathbb{R}^m$ and $\boldsymbol{\theta} \in \Theta \subseteq \mathbb{R}^p$, where $q(\mathbf{x};\boldsymbol{\theta})$ is analytically tractable but $Z(\boldsymbol{\theta}) :\, = \int q(\mathbf{x};\boldsymbol{\theta}) d\mathbf{x}$ is computationally intractable. Throughout this article, we say that the model is {\it correctly specified} if there exists $\boldsymbol{\theta}_0 \in \Theta$ such that $p_{\mathbf{x}}(\cdot) = p(\cdot; \boldsymbol{\theta}_0)$. Otherwise, the model is {\it misspecified}. To address the intractable normalizing constant in $p(\mathbf{x};\boldsymbol{\theta})$, \citet{hyvarinen2005estimation} introduced score matching as an estimation procedure that bypasses the computation of $Z(\boldsymbol{\theta})$. Specifically, the score function for the model density is denoted as $\psi(\mathbf{x}; \boldsymbol{\theta}) = \nabla_{\mathbf{x}} \log p(\mathbf{x}; \boldsymbol{\theta}) = \nabla_{\mathbf{x}} \log q(\mathbf{x}; \boldsymbol{\theta})$, where $\nabla_{\mathbf{x}}$ denotes the gradient with respect to the data vector $\mathbf{x}$. Next, the score function associated with the data distribution is defined as $\psi_{\mathbf{x}}(\cdot) = \nabla_{\mathbf{x}} \log p_{\mathbf{x}}(\cdot)$, which is unknown. The target parameter of score matching is then defined as the minimizer of the population score matching loss $J(\cdot)$, 
\begin{equation}
\label{theta^*}
    \boldsymbol{\theta}^* = \argmin_{\boldsymbol{\theta} \in \Theta} J(\boldsymbol{\theta}), \text{ where }  J(\boldsymbol{\theta}) = \frac{1}{2} \int \|\psi(\mathbf{x}; 
\boldsymbol{\theta}) - \psi_{\mathbf{x}}(\mathbf{x})\|^2 p_{\mathbf{x}}(d\mathbf{x}),
\end{equation}
with $\|\cdot\|$ the Euclidean norm. The population loss $J(\cdot)$ is sometimes called the relative Fisher divergence \citep{lyu2012interpretation}. Importantly, since the loss depends on the model density only through the gradient of its log-density, it suffices to know the model density only up to a normalizing constant. Under mild regularity conditions \citep{hyvarinen2005estimation}, $J(\boldsymbol{\theta})$ can be written as
\begin{equation}
\label{sm_2}
    J(\boldsymbol{\theta}) = \mathbb{E}_{p_{\mathbf{x}}} \bigg[\mbox{tr} \big(\nabla_{\mathbf{x}}^2 \log p(\mathbf{x}; \boldsymbol{\theta}) \big) + \frac{1}{2} \|\nabla_{\mathbf{x}} \log p(\mathbf{x}; \boldsymbol{\theta})\|^2 \bigg] + \text{const}, 
\end{equation}
where $\mathbb{E}_{p_{\mathbf{x}}}$ denotes expectation with respect to $p_{\mathbf{x}}$, $\mbox{tr}(A)$ denotes the trace of the matrix $A$, $\nabla_{\mathbf{x}}^2$ denotes the Hessian with respect to the data vector $\mathbf{x}$, and $\text{const}$ denotes a constant that does not depend on $\boldsymbol{\theta}$. For independent samples $\mathbf{x}_{1:n} = (\mathbf{x}_1, \mathbf{x}_2, \ldots, \mathbf{x}_n)$ of size $n$, the corresponding sample version of $J(\boldsymbol{\theta})$ is
\begin{equation}
\label{sm_sampleversion}
\widehat{J}(\boldsymbol{\theta}; \mathbf{x}_{1:n}) = \frac{1}{n}\sum_{i=1}^n \sum_{j=1}^m [\partial_j \psi_j(\mathbf{x}_i; \boldsymbol{\theta}) + \frac{1}{2} \psi_j^2(\mathbf{x}_i; \boldsymbol{\theta})] + \text{const},
\end{equation}
where $\psi_j(\mathbf{x}; \boldsymbol{\theta})$ is the $j$-th element of $\psi(\mathbf{x}; \boldsymbol{\theta})$ and $\partial_j \psi_j(\mathbf{x}; \boldsymbol{\theta})$ is the partial derivative of $\psi_j(\mathbf{x}; \boldsymbol{\theta})$ with respect to the $j$-th variable in the data vector. We call the minimizer of $\widehat{J}(\boldsymbol{\theta})$ the {\it sample score matching estimator} in the sequel. 

\subsection{Methodology}
\label{sec:methodology}

The score matching target parameter $\boldsymbol{\theta}^*$ in \eqref{theta^*} is defined as the minimizer of the population criterion $J(\boldsymbol{\theta})$. To use this criterion within ETEL, we work with the first-order optimality condition of this minimization problem. When $\boldsymbol{\theta}^*$ is an interior minimizer and differentiation can be interchanged with expectation, $\boldsymbol{\theta}^*$ satisfies the population estimating equation
\begin{equation}
\label{population_moment_condition}
\mathbb{E}_{p_{\mathbf{x}}}
\left\{
\mathbf{g}(\mathbf{x}, \boldsymbol{\theta}^*)
\right\}
=
\mathbf{0}_p,
\end{equation}
where $\mathbf{0}_p$ is the $p \times 1$ vector of zeros, and $\mathbf{g}(\mathbf{x}, \boldsymbol{\theta}) :\, = \nabla_{\boldsymbol{\theta}}\left[\mbox{tr}\big(
\nabla_{\mathbf{x}}^2 \log p(\mathbf{x}; \boldsymbol{\theta})\big) + 2^{-1}\left\|\nabla_{\mathbf{x}} \log p(\mathbf{x}; \boldsymbol{\theta})\right\|^2\right]$. Motivated by this estimating equation, we propose \textbf{SME-BETEL}, which uses the first-order optimality conditions of the score matching objective as moment conditions in Bayesian ETEL, thereby enabling robust inference for doubly-intractable models.

The equation above defines a population-level moment condition. Given independent samples $\mathbf{x}_{1:n}$, we can construct the corresponding sample-specific moment functions $\mathbf{g}(\mathbf{x}_i, \boldsymbol{\theta})$ for $i = 1, \ldots, n$. The ETEL function $L: \Theta \to (0, \infty)$ is then constructed by solving a constrained optimization problem \citep{schennach2005bayesian}. This problem is solvable only when the origin lies in the interior of the convex hull $\mathcal{C}(\boldsymbol{\theta}) :\, = \left\{\sum_{i=1}^n \omega_i \mathbf{g}(\mathbf{x}_i, \boldsymbol{\theta}) : \omega_i \geq 0,\ \sum_{i=1}^n \omega_i = 1\right\}$ of $\{\mathbf{g}(\mathbf{x}_i, \boldsymbol{\theta})\}_{i=1}^n$. The recent regularized ETEL offers an approach to relaxing this convex-hull restriction \citep{kim2023regularized}. Here, we focus on the standard ETEL framework, for which the ETEL function is defined as:
\begin{equation}
\label{likelihood}
L(\boldsymbol{\theta}) = \mathbbm{1}_{\mathcal{C}(\boldsymbol{\theta})} ( \mathbf{0}) \,  \prod_{i=1}^n \omega_i(\boldsymbol{\theta}), 
\end{equation}
where $\mathbbm{1}_{\mathcal{C}(\boldsymbol{\theta})}$ denotes the indicator function of a set $\mathcal{C}(\boldsymbol{\theta})$, and $\boldsymbol{\omega}(\boldsymbol{\theta})\,:= \big(\omega_1(\boldsymbol{\theta}), \ldots, \omega_n(\boldsymbol{\theta})\big)$ solves the following constrained optimization problem:
\begin{equation}
\label{optimization_problem}
\max_{ \boldsymbol{\omega}(\boldsymbol{\theta})  \in \mathcal{S}^{n-1}} \quad 
\sum_{i=1}^n \bigl[-\, \omega_i \log(n \omega_i)\bigr] 
\, \text{ subject to }  
\sum_{i=1}^n \omega_i\, \mathbf{g}(\mathbf{x}_i,\boldsymbol{\theta}) = \mathbf{0},
\end{equation}
where $\mathcal{S}^{n-1} := \{v \in \mathbb{R}^n : v_i \geq 0,\ i = 1, \ldots, n,\ \sum_{i=1}^n v_i = 1\}$ denotes the $(n-1)$-dimensional unit simplex. By introducing Lagrange multipliers to the constraints, these probabilities $\boldsymbol{\omega}(\boldsymbol{\theta})$ can be equivalently expressed as \citep{schennach2005bayesian}
\begin{equation*}
\omega_i(\boldsymbol{\theta})
= \frac{\exp \bigl(\widehat{\boldsymbol{\lambda}}(\boldsymbol{\theta})^{\top} \mathbf{g}(\mathbf{x}_i,\boldsymbol{\theta})\bigr)}
       {\sum_{j=1}^n \exp \bigl(\widehat{\boldsymbol{\lambda}}(\boldsymbol{\theta})^{\top} \mathbf{g}(\mathbf{x}_j,\boldsymbol{\theta})\bigr)}, \text{ where } \widehat{\boldsymbol{\lambda}}(\boldsymbol{\theta}) = \argmin_{\boldsymbol{\lambda} \in \mathbb{R}^p} \frac{1}{n}\sum_{i=1}^n \exp \big(\boldsymbol{\lambda}^{\top}\mathbf{g}(\mathbf{x}_i, \boldsymbol{\theta})\big).
\end{equation*}
It is worth mentioning that $\boldsymbol{\omega}(\boldsymbol{\theta})$ can be interpreted as probability weights assigned to the sample observations that are closest to the empirical weights $(n^{-1},\ldots,n^{-1})$ in Kullback–Leibler (KL) divergence, while satisfying the weighted sample moment condition.

In a Bayesian framework, the ETEL function $L(\boldsymbol{\theta})$ serves as a surrogate for the conventional likelihood, yielding the calibration-free SME-BETEL posterior density:
\begin{equation}
\label{SME-BETEL_posterior}
    \pi(\boldsymbol{\theta} \mid \mathbf{x}_{1:n}) = \frac{L(\boldsymbol{\theta})\pi(\boldsymbol{\theta})}{\int_{\Theta} L(\boldsymbol{\eta}) \pi(\boldsymbol{\eta}) d\boldsymbol{\eta}},
\end{equation}
where $\pi(\boldsymbol{\theta})$ is the prior distribution. This posterior distribution can be efficiently explored using the MH algorithm. An important consideration during sampling is that any proposed parameter for which the origin does not lie in the interior of the convex hull is discarded. Algorithm~\ref{alg:mh-sme-betel} summarizes the MH sampler for the SME-BETEL posterior. The construction above was presented for data supported on $\mathbb{R}^m$, where the original score matching objective applies directly. Many doubly-intractable models, however, are defined on more general supports, such as truncated domains or mixed-domain supports involving both Euclidean and bounded components. This motivates support-adapted versions of SME-BETEL, which we describe next.

\subsection{SME-BETEL for Mixed-Domain Data}
\label{sec:support_adapted_sm}

The flexibility of SME-BETEL comes from the modular role of score matching. The original score matching objective applies to data supported on $\mathbb{R}^m$, but alternative score matching objectives have been developed fro nonstandard supports, including nonnegative, compact, and manifold spaces \citep{hyvarinen2007some, mardia2016score, yu2019generalized,liu2022estimating,scealy2023score}. Whenever an appropriate score matching objective is available, its first-order conditions can be used as moment conditions in the SME-BETEL. Existing support-adapted score matching methods have largely been developed for individual support types. Many applications, however, involve observations whose components lie on different supports. We therefore focus on extending SME-BETEL to such mixed-domain settings.

Examples of mixed-domain data arise in preferential sampling \citep{gelfand2012effect} and marked spatial point processes \citep{eckardt2023marked}, where Euclidean-valued responses or marks are observed together with locations on a compact domain. Motivated by these settings, we consider observations $\mathbf{x}=(\mathbf{y},\mathbf{s})$, where $\mathbf{y}\in\mathbb{R}^d$ and $\mathbf{s}\in\mathcal{D}$ is a compact domain. We write the data-generating density as $p_{\mathbf{x}}(\mathbf{y},\mathbf{s})
= p_{\mathbf{x}}(\mathbf{y}\mid\mathbf{s})p_{\mathbf{x}}(\mathbf{s})$ and consider a working model with the compatible factorization $p(\mathbf{y},\mathbf{s};\boldsymbol{\theta}) = p(\mathbf{y}\mid\mathbf{s};\boldsymbol{\theta})p(\mathbf{s};\boldsymbol{\theta})$. This conditional-marginal structure provides a general route for extending support-adapted score matching to other mixed-domain models.

Our construction starts from the original relative Fisher divergence in \eqref{theta^*} and adapts it to the mixed support. For observations with Euclidean support, the usual score matching criterion can be used directly. In the present setting, however, the location component lies on the compact domain $\mathcal{D}$, so the score discrepancy must be modified near the boundary to obtain a well-behaved computable criterion. We therefore introduce a smooth weight function $h:\mathcal{D}\to\mathbb{R}$ that equals one over most of the interior of $\mathcal{D}$ and tapers smoothly to zero near the boundary. The specific weight function used in our numerical studies is given in Appendix~\ref{subsec:weight_function}. This choice preserves the usual relative Fisher divergence away from the boundary, where most inferential information is retained, while downweighting only the boundary region where compact-support effects arise. Consequently, the resulting mixed-domain criterion remains a valid nonnegative discrepancy between the model score and the data-generating score, and it is minimized at the data-generating distribution under the usual identifiability conditions. A technical contribution of this paper is to combine these ingredients into the following {\it mixed-domain score matching objective}:
\begin{equation}
\label{sm_mixed_type_defnition}
\begin{aligned}
J_{MD}(\boldsymbol{\theta})
&=
\frac{1}{2}
\int_{\mathbb{R}^d \times \mathcal{D}}
h(\mathbf{s})
\left\|
\psi(\mathbf{y}, \mathbf{s}; \boldsymbol{\theta})
-
\psi_{\mathbf{x}}(\mathbf{y}, \mathbf{s})
\right\|^2
p_{\mathbf{x}}(\mathbf{y}, \mathbf{s})\,
\nu(d\mathbf{y})\mu(d\mathbf{s}),
\end{aligned}
\end{equation}
where $\psi(\mathbf{y},\mathbf{s};\boldsymbol{\theta}) = \nabla_{(\mathbf{y},\mathbf{s})}\log p(\mathbf{y},\mathbf{s};\boldsymbol{\theta})$ and $\psi_{\mathbf{x}}(\mathbf{y},\mathbf{s}) = \nabla_{(\mathbf{y},\mathbf{s})}\log p_{\mathbf{x}}(\mathbf{y},\mathbf{s})$. The following proposition gives an equivalent computable form of \eqref{sm_mixed_type_defnition}.

\begin{proposition}
\label{prop:mixed_domain_sm}
The mixed-domain score matching objective function \eqref{sm_mixed_type_defnition} admits the following alternative formula:
\begin{equation}
\label{sm_mixed_type}
\begin{aligned}
J_{MD}(\boldsymbol{\theta})
&=
\mathbb{E}_{p_{\mathbf{x}}}\bigg[
\frac{1}{2}h(\mathbf{s})
\left\|
\psi(\mathbf{y}, \mathbf{s}; \boldsymbol{\theta})
\right\|^2
+
h(\mathbf{s})
\operatorname{tr}
\left\{
\nabla^2_{\mathbf{y}}\log p(\mathbf{y}\mid \mathbf{s};\boldsymbol{\theta})
\right\} \\
&\qquad
+
h(\mathbf{s})
\Delta_{\mathbf{s}} \log p(\mathbf{y}, \mathbf{s};\boldsymbol{\theta})
+
\nabla_{\mathbf{s}}h(\mathbf{s})^{\top}
\nabla_{\mathbf{s}}\log p(\mathbf{y}, \mathbf{s};\boldsymbol{\theta})
\bigg]
+
\operatorname{const},
\end{aligned}
\end{equation}
where $\Delta_{\mathbf{s}}$ denotes the Laplacian operator with respect to $\mathbf{s}$, and $\operatorname{const}$ is independent of $\boldsymbol{\theta}$.
\end{proposition}
The proof is provided in Appendix~\ref{app:heterogeneous_SM_derivation}. Proposition~\ref{prop:mixed_domain_sm} is useful for practical implementation because the right hand side of \eqref{sm_mixed_type} does not involve the unknown score function $\psi_{\mathbf{x}}$. Hence, it yields a computable sample criterion, and differentiating \eqref{sm_mixed_type} with respect to $\boldsymbol{\theta}$ gives the estimating equations used by SME-BETEL for mixed-domain data. Importantly, this construction also eliminates the need to evaluate the spatial normalizing constant associated with the location density. For example, preferential sampling models often specify $p(\mathbf{s};\boldsymbol{\theta}) = \exp\{\xi(\mathbf{s};\boldsymbol{\theta})\}
/ \int_{\mathcal{D}}\exp\{\xi(\mathbf{u};\boldsymbol{\theta})\}\,d\mathbf{u}$. Since the denominator is constant as a function of $\mathbf{s}$, we have $\nabla_{\mathbf{s}}\log p(\mathbf{s};\boldsymbol{\theta})
= \nabla_{\mathbf{s}}\xi(\mathbf{s};\boldsymbol{\theta})$ and $\Delta_{\mathbf{s}}\log p(\mathbf{s};\boldsymbol{\theta})
= \Delta_{\mathbf{s}}\xi(\mathbf{s};\boldsymbol{\theta})$. Thus, the spatial normalizing integral does not appear in the moment equations. This feature is central to the preferential sampling simulation in Section~\ref{sec:simulation_study} and the ozone-monitoring application in Section~\ref{real_data}.

\section{Simulation Study}
\label{sec:simulation_study}

In this section, we evaluate the performance of the mixed-domain score matching objective in \eqref{sm_mixed_type} within the SME-BETEL framework. We focus on a preferential sampling problem, where each observation consists of a response-location pair $(y, \mathbf{s})$, with $y \in \mathbb{R}$ and $\mathbf{s} \in \mathcal{D} \subseteq \mathbb{R}^2$. Preferential sampling arises when the sampling locations are not selected independently of the underlying spatial process. In environmental monitoring, for example, monitoring sites may be preferentially located in regions where the response is expected to take large values \citep{gelfand2012effect}. In this case, the observed locations themselves contain information about the latent spatial surface, and ignoring this dependence can lead to biased inference. A common modeling strategy links the sampling density to a latent spatial process through an intensity of the form $p(\mathbf{s}) \propto \exp\{\xi(\mathbf{s})\}$, whose normalizing constant is an integral over the spatial domain. This makes preferential sampling models a practically important class of mixed-domain doubly-intractable models for which SME-BETEL is well suited. Mixed-domain score matching avoids the spatial normalizing integral, while Bayesian ETEL constructs a semiparametric posterior from the resulting estimating equations. To assess robustness under misspecification, we compare SME-BETEL with the likelihood-based Bayesian approach of \citet{pati2011bayesian}, which we refer to as the {\it PRD} method.

\paragraph*{Misspecified Preferential Sampling Model} We first specify the preferential sampling model used as the working model. Let $\mu(\mathbf{s})$ denote the spatial surface at location $\mathbf{s}\in\mathcal{D}=[0,1]^2$. Following the preferential sampling framework of \citet{pati2011bayesian}, for $i=1,\ldots,n$, the working model jointly specifies the response process and the sampling-location process as $y_i \mid \mathbf{s}_i \sim N (\mu(\mathbf{s}_i) = \eta(\mathbf{s}_i) + a \xi(\mathbf{s}_i), \sigma^2)$ and $p(\mathbf{s}_i) = \exp\{\xi(\mathbf{s}_i)\} / \int_{\mathcal{D}} \exp\{\xi(\mathbf{u})\} d\mathbf{u}$. Here $\eta(\mathbf{s})$ represents a baseline spatial surface, while $\xi(\mathbf{s})$ controls the spatial density of the monitoring locations. The scalar parameter $a$ measures the association between the sampling intensity and the response surface: when $a=0$, the locations are noninformative for the response, whereas nonzero values of $a$ indicate informative sampling. Letting $x(\mathbf{s})$ denote a vector of spatial covariates, we write $\xi(\mathbf{s}) = x(\mathbf{s})^\top \boldsymbol{\beta}_{\xi} + \xi_r(\mathbf{s})$ and $\eta(\mathbf{s}) = x(\mathbf{s})^\top \boldsymbol{\beta}_{\eta} + \eta_r(\mathbf{s})$, where $\boldsymbol{\beta}_{\xi}$ and $\boldsymbol{\beta}_{\eta}$ are regression coefficients and $\xi_r(\mathbf{s})$ and $\eta_r(\mathbf{s})$ are zero-mean residual processes. 

We assign independent zero-mean Gaussian process priors to $\xi_r$ and $\eta_r$, each with $\mathrm{Mat\acute{e}rn}$ covariance function: $c(r \mid \boldsymbol{\psi}) = \tau^2/(2^{\nu-1}\Gamma(\nu)) \times (2\nu^{1/2}r/\rho)^{\nu} \times \mathcal{K}_{\nu}(2\nu^{1/2}r/\rho)$ and $r = ||\mathbf{s} - \mathbf{s}'||$, where $\boldsymbol{\psi}=(\tau^2,\rho,\nu)$, $\mathcal{K}$ is the modified Bessel function of the second kind, $\tau^2$ controls variance, $\rho$ controls spatial range, and $\nu$ controls smoothness.  For computation, we replace the full Gaussian processes by low-rank kernel convolution approximations \citep{higdon2002space}. Specifically, if $\delta(\mathbf{s})$ is a zero-mean Gaussian process with covariance $c(r\mid\boldsymbol{\psi})$, we approximate it using spatial knots $\boldsymbol{\phi}_1,\ldots,\boldsymbol{\phi}_N$ as $\delta(\mathbf{s}) \approx \sum_{j=1}^N K_{\boldsymbol{\psi}}(\mathbf{s} - \boldsymbol{\phi}_j) w_j$ for sufficiently large $N$, where $K_{\boldsymbol{\psi}}$ is the kernel associated with $c(r\mid\boldsymbol{\psi})$ and $w_j \sim N(0,1)$. The specific form of $K_{\boldsymbol{\psi}}$ is given in Appendix~\ref{subsec:kernel_form}. Applying this representation to $\xi(\mathbf{s})$ and $\eta(\mathbf{s})$ yields
\begin{equation}
\begin{aligned}
\label{full_IS}
y_i \mid \mathbf{s}_i &\sim \mathcal{N} \bigg( x(\mathbf{s}_i)^\top (a\boldsymbol{\beta}_{\xi} + \boldsymbol{\beta}_{\eta}) + \sum_{j=1}^N K_{\boldsymbol{\psi}_{\eta}}(\mathbf{s}_i - \boldsymbol{\phi}_j) u_j + a \sum_{j=1}^N K_{\boldsymbol{\psi}_{\xi}}(\mathbf{s}_i - \boldsymbol{\phi}_j) v_j, \, \sigma^2 \bigg), \\
p(\mathbf{s}_i) &= \frac{\exp\left\{ x(\mathbf{s}_i)^\top \boldsymbol{\beta}_{\xi} + \sum_{j=1}^N K_{\boldsymbol{\psi}_{\xi}}(\mathbf{s}_i - \boldsymbol{\phi}_j) v_j \right\}}{\int_{\mathcal{D}} \exp\left\{ x(\mathbf{u})^\top \boldsymbol{\beta}_{\xi} + \sum_{j=1}^N K_{\boldsymbol{\psi}_{\xi}}(\mathbf{u} - \boldsymbol{\phi}_j) v_j \right\} d\mathbf{u}},
\end{aligned}
\end{equation}
where $u_j, v_j \sim N(0, 1)$. Our main inferential target is the latent spatial surface $\mu(\mathbf{s})$, evaluated over a grid of spatial locations. To do this, we estimate $(\boldsymbol{\beta}_{\xi}, \boldsymbol{\beta}_{\eta}, a, \sigma^2, \tau^2_{\xi}, \rho_{\xi}, \nu_{\xi}, \tau^2_{\eta}, \rho_{\eta}, \nu_{\eta}, u_1, \ldots, u_N, v_1, \ldots, v_N)$. We set $x(\mathbf{s}) = 1$ for all $\mathbf{s}$. Then, we assume that the data-generating model is
\begin{equation*}
\begin{aligned}
y_i \mid \mathbf{s}_i &\sim t_5^{\mathrm{nc}} \bigg( x(\mathbf{s}_i)^\top (a\boldsymbol{\beta}_{\xi} + \boldsymbol{\beta}_{\eta}) + \sum_{j=1}^N K_{\boldsymbol{\psi}_{\eta}}(\mathbf{s}_i - \boldsymbol{\phi}_j) u_j + a \sum_{j=1}^N K_{\boldsymbol{\psi}_{\xi}}(\mathbf{s}_i - \boldsymbol{\phi}_j) v_j \bigg), \\
p(\mathbf{s}_i) &= \frac{\exp\left\{ x(\mathbf{s}_i)^\top \boldsymbol{\beta}_{\xi} + \sum_{j=1}^N K_{\boldsymbol{\psi}_{\xi}}(\mathbf{s}_i - \boldsymbol{\phi}_j) v_j \right\}}{\int_{\mathcal{D}} \exp\left\{ x(\mathbf{u})^\top \boldsymbol{\beta}_{\xi} + \sum_{j=1}^N K_{\boldsymbol{\psi}_{\xi}}(\mathbf{u} - \boldsymbol{\phi}_j) v_j \right\} d\mathbf{u}},
\end{aligned}
\end{equation*}
Here $t^{\mathrm{nc}}_5(\mu(\mathbf{s}_i))$ denotes a noncentral Student $t$ distribution with 5 degrees of freedom and noncentrality parameter $\mu(\mathbf{s}_i)$. The working model is given by \eqref{full_IS}. Misspecification arises because the data-generating response distribution is heavy-tailed, whereas the working model assumes Gaussian errors. The PRD method requires numerical approximation of the normalizing integral in $p(\mathbf{s}_i)$; we use the same integration grid of size $M=50^2$ to generate sampling locations from the data-generating model. To make the two methods comparable, we fix $\rho_{\xi} = \rho_{\eta} = 1$ and $\tau^2_{\xi} = \tau_{\eta}^2 = 0.1^2$ throughout the simulation study. For $\nu_{\xi}$ and $\nu_{\eta}$, because the derivative of $K_{\boldsymbol{\psi}}(\boldsymbol{u})$ with respect to $\nu$ is not available in closed form, we fix $\nu_{\xi} = \nu_{\eta} = 6$ in both methods. Motivated by \citet{rodrigues2010class}, we use an equally spaced $N = 6\times 6$ grid of knots on $[0.01, 0.99]^2$ for the kernel convolution approximation. 
For the remaining parameters, we specify the priors $a \sim \mathrm{N}(0, 10^2)$, $\sigma^2 \sim \mathrm{Inv\text{-}Ga}(5, 5)$, $\beta_{\xi} \sim \mathrm{N}(0, 5)$, and $\beta_{\eta} \sim \mathrm{N}(0, 5)$. 

For this experiment, we apply the new score matching method defined for mixed-domain data in \eqref{sm_mixed_type} and use the weight function specified in \eqref{mollifier} with $\epsilon=0.2$. We use MH sampling to generate 5{,}000 posterior samples and discard the first 1{,}000 as burn-in. The proposal scale is tuned so that the acceptance rate is approximately 0.4. We conduct 50 replications under two simulation scenarios: (i) $n=250$; and (ii) $n=400$. For each dataset, we compare the predicted surface over the grid with the latent signal surface for each model and report the bias, mean squared error (MSE), and mean absolute deviation (MAD), averaged over the grid of spatial locations. The results are reported in Table~\ref{tab:simres_informative_sampling}. As shown, the bias, MSE, and MAD of SME-BETEL are lower than those of the PRD method across both cases. As the sample size increases, the MSE and MAD decrease for all methods, yet SME-BETEL consistently attains smaller values. These results indicate that SME-BETEL provides a more accurate representation of the latent signal surface in this mixed-domain doubly-intractable model, illustrating its practical robustness under model misspecification. We next study the large-sample properties of the core SME-BETEL construction based on the original score matching objective in an i.i.d. continuous domain setting.

\begin{table}[!t]
\centering
\caption{Simulation results for the preferential sampling model under model misspecification.}
\label{tab:simres_informative_sampling}
\begingroup
\small
\setlength{\tabcolsep}{4pt}
\renewcommand{\arraystretch}{0.92}
\begin{tabular}{@{}llrrr@{}}
\toprule
Sample Size & Model & Bias $(\times 10^2)$ & MSE $(\times 10^2)$ & MAD $(\times 10^2)$ \\
\midrule
\multirow{2}{*}{$n = 250$}
  & PRD method & 22.9 & 10.7 & 26.4 \\
  & SME-BETEL  & \textbf{15.3} & \textbf{9.4} & \textbf{24.4} \\
\addlinespace[1pt]
\multirow{2}{*}{$n = 400$}
  & PRD method & 22.7 & 8.7 & 24.6 \\
  & SME-BETEL  & \textbf{15.3} & \textbf{6.2} & \textbf{20.5} \\
\bottomrule
\end{tabular}
\endgroup
\end{table}

\section{Theoretical Properties of SME-BETEL}
\label{sec:sme_betel}
Having demonstrated the performance of SME-BETEL in a mixed-domain doubly-intractable problem, we now study the theoretical properties of the core SME-BETEL construction. The results in this section are developed for the original score matching objective in \eqref{sm_2} under an i.i.d. setting with observations supported on $\mathbb{R}^m$. Building on the general framework for Bayesian ETEL posteriors in \citet{chib2018bayesian}, we establish a BvM theorem for the SME-BETEL posterior defined in \eqref{SME-BETEL_posterior}. Before presenting the theorem, we introduce some additional notations and state our assumptions. 
Define
\begin{equation}
\label{ll_one_observation}
    \ell_{n, \boldsymbol{\theta}}(\mathbf{x}) := \log \frac{\exp \bigl(\widehat{\boldsymbol{\lambda}}(\boldsymbol{\theta})^{\top} \mathbf{g}(\mathbf{x},\boldsymbol{\theta})\bigr)}
       {\sum_{j=1}^n \exp \bigl(\widehat{\boldsymbol{\lambda}}(\boldsymbol{\theta})^{\top} \mathbf{g}(\mathbf{x}_j,\boldsymbol{\theta})\bigr)},
\end{equation}
so that the log-ETEL function is $\sum_{i=1}^n \ell_{n, \boldsymbol{\theta}}(\mathbf{x}_i)$. We impose the following assumptions:
\begin{assumption}
\label{s_1}
    The prior admits a density $\pi(\boldsymbol{\theta})$ with respect to Lebesgue measure, and $\pi(\boldsymbol{\theta})$ is continuous and positive at $\boldsymbol{\theta}^*$, where $\boldsymbol{\theta}^*$ is defined in \eqref{theta^*}.
\end{assumption}
\begin{assumption}
\label{s_2}
    The moment condition $\mathbb{E}_{p_{\mathbf{x}}}[\mathbf{g}(\mathbf{x}, \boldsymbol{\theta})] = \mathbf{0}$ has a unique solution $\boldsymbol{\theta}^* \in\operatorname{int}(\Theta)$, where $\mathbf{g}(\mathbf{x}, \boldsymbol{\theta})$ is defined as in \eqref{population_moment_condition}. Moreover $\Theta$ is compact.
\end{assumption}
\begin{assumption}
\label{s_3}
     (a) The observations $\mathbf{x}_{1:n}$ are i.i.d. with common density $p_{\mathbf{x}}$; (b) $\mathbf{g}(\mathbf{x}, \boldsymbol{\theta})$ is continuous at each $\boldsymbol{\theta} \in \Theta$ with probability one; (c) $\mathbf{g}(\mathbf{x}, \boldsymbol{\theta})$ is continuously differentiable in an open neighborhood $\mathcal{N}$ of $\boldsymbol{\theta}^*$; (d) $\mathbb{E}_{p_{\mathbf{x}}}[\sup_{\boldsymbol{\theta} \in \Theta}\|\mathbf{g}(\mathbf{x}, \boldsymbol{\theta})\|^{\alpha}] < \infty$ for some $\alpha > 2$; (e) $\mathbb{E}_{p_{\mathbf{x}}}[\sup_{\boldsymbol{\theta} \in \mathcal{N}} \|\partial \mathbf{g}(\mathbf{x}, \boldsymbol{\theta}) / \partial \boldsymbol{\theta}^{\top}\|_{F}] < \infty$ where $\|\cdot\|_F$ denotes the Frobenius norm.
\end{assumption}
\begin{assumption}
\label{s_4}
    (a) $\Delta :\,= \mathbb{E}_{p_{\mathbf{x}}}[\mathbf{g}(\mathbf{x}, \boldsymbol{\theta}^*)\mathbf{g}(\mathbf{x}, \boldsymbol{\theta}^*)^{\top}]$ is non-singular; (b) $\Gamma :\,= \mathbb{E}_{p_{\mathbf{x}}}[\partial \mathbf{g}(\mathbf{x}, \boldsymbol{\theta}^*)/\partial \boldsymbol{\theta}^{\top}]$ and $\operatorname{rank}(\Gamma)=p$.
\end{assumption}
\begin{assumption}
\label{bvm_condition}
    For any $\delta > 0$, there exists an $\epsilon > 0$ such that, as $n \to \infty$,
    \begin{equation*}
        P \bigg(\sup_{\|\boldsymbol{\theta} - \boldsymbol{\theta}^*\| > \delta} 
        \frac{1}{n}\sum_{i=1}^n 
        \big(\ell_{n, \boldsymbol{\theta}}(\mathbf{x}_i) 
        - \ell_{n, \boldsymbol{\theta}^*}(\mathbf{x}_i)\big) 
        \leq -\epsilon \bigg) \to 1,
    \end{equation*}
    where $\ell_{n, \boldsymbol{\theta}}(\mathbf{x})$ is defined in \eqref{ll_one_observation} and P denotes the probability under $p_{\mathbf{x}}$.
\end{assumption}
Assumption~\ref{s_1} concerns the prior distribution, while Assumption~\ref{bvm_condition} ensures that the log-ETEL criterion is sufficiently small whenever $\boldsymbol{\theta}$ is far from $\boldsymbol{\theta}^*$. These are standard conditions used to establish the asymptotic properties of Bayesian procedures \citep{van2000asymptotic, kleijn2012bernstein, chib2018bayesian}. Assumptions~\ref{s_2}–\ref{s_4} are the same as those imposed in \citet{newey2004higher, schennach2007point, chib2018bayesian}. The BvM theorem is stated below; we use $\mathcal{N}(A; \mu, \Sigma)$ to denote $P(Z \in A)$ where $Z \sim N(\mu, \Sigma)$. 
\begin{theorem}[BvM for SME-BETEL]
\label{bvm}
    Under Assumptions~\ref{s_1}--\ref{bvm_condition}, the SME-BETEL posterior defined in \eqref{SME-BETEL_posterior} satisfies
    \begin{equation}
    \label{target}
        \sup_{B} \bigg| 
        \pi\left(\sqrt{n}(\boldsymbol{\theta} - \widehat{\boldsymbol{\theta}}_n) \in B 
        \mid \mathbf{x}_{1:n}\right) 
        - \mathcal{N}(B; {\mathbf{0}, (\Gamma^{\top}\Delta^{-1}\Gamma)^{-1}}) 
        \bigg| \overset{p}{\to} 0,
    \end{equation}
    where $\Delta$ and $\Gamma$ are defined in Assumptions~\ref{s_4}, and 
    $\widehat{\boldsymbol{\theta}}_n$ is the sample score matching estimator, i.e., $\widehat{\boldsymbol{\theta}}_n = \argmin_{\theta \in \Theta} \widehat{J}(\boldsymbol{\theta}; \mathbf{x}_{1:n})$. In \eqref{target}, the supremum is taken over all Borel sets $B \subseteq \Theta$ and $\convp$ indicates convergence in probability with respect to the true data law. 
\end{theorem}
\begin{proof}
    The proof can be found in Appendix~\ref{app:bvm}.
\end{proof}
The above theorem implies that the total variation (TV) distance between the SME-BETEL posterior $\pi(\boldsymbol{\theta} \mid \mathbf{x}_{1:n})$, and a Gaussian distribution centered at the sample score matching estimator $\widehat{\boldsymbol{\theta}}_n$ with variance $n^{-1}(\Gamma^{\top}\Delta^{-1}\Gamma)^{-1}$, converges in probability to zero. The properties of $\widehat{\boldsymbol{\theta}}_n$ have been studied: \citet{hyvarinen2005estimation} established convergence of the sample criterion in \eqref{sm_sampleversion} to its population counterpart; \citet{forbes2015linear} established consistency and asymptotic normality of $\widehat{\boldsymbol{\theta}}_n$ for exponential-family models; asymptotic normality has also been studied for a variant of score matching, namely sliced score matching \citep{song2020sliced}. Our Theorem~\ref{asymptotic_normality} extends this line of work by establishing asymptotic normality of $\widehat{\boldsymbol{\theta}}_n$ under the general setting considered in this
paper, where we show that
\begin{equation}
\label{asymptotic_normal}
    \sqrt{n}(\widehat{\boldsymbol{\theta}}_n - \boldsymbol{\theta}^*) \convd 
    N\big(\mathbf{0}, (\Gamma^{\top}\Delta^{-1}\Gamma)^{-1}\big),
\end{equation}
where $\boldsymbol{\theta}^*$ is defined in \eqref{theta^*} and $\convd$ denotes convergence in distribution. Importantly, equations~\eqref{target} and~\eqref{asymptotic_normal} show that the limiting posterior variance of SME-BETEL matches the asymptotic sampling variance of the score matching estimator under both correct specification and misspecification. Consequently, SME-BETEL credible sets are also asymptotically valid frequentist confidence sets for $\boldsymbol{\theta}^*$. This calibration is particularly important under model misspecification, where the usual Bayesian posterior does not generally provide valid frequentist coverage \citep{kleijn2012bernstein}. Appendix~\ref{subsec:calibration_example} provides an illustrative example demonstrating this distinction. 

Under correct specification, both SME-BETEL and the usual Bayesian posterior provide asymptotically valid uncertainty quantification, although SME-BETEL can be less statistically efficient because score matching does not use the information contained in the normalizing constant. This efficiency loss is consistent with existing results on the statistical efficiency of score matching \citep{koehler2022statistical}. Appendix~\ref{app:simulation_res} provides extensive simulation studies across different models and data supports, illustrating the accuracy of SME-BETEL under correct specification and its uncertainty calibration under model misspecification.

\section{Eastern United States Ozone Data}
\label{real_data}

Having established the theoretical properties of the core SME-BETEL construction, we now return to the mixed-domain preferential sampling setting in Section~\ref{sec:simulation_study}. As a real-data counterpart to the preferential sampling simulation, we re-analyze the Eastern United States ozone data in \cite{pati2011bayesian} to evaluate SME-BETEL in a spatial setting where the sampling design may depend on the underlying response surface. The data consist of median daily ozone measurements from June–August 2007 at $n=631$ monitoring locations. Because monitoring sites are not placed uniformly across the region, the observed spatial pattern may contain information about the underlying ozone process. This makes the dataset useful for assessing whether accounting for informative sampling changes inference results. The left plot in Figure~\ref{fig:ozone_data_combined} displays the observed ozone levels and monitor locations. The monitoring network is denser in several urban and northeastern regions, including New York and Connecticut, where ozone levels tend to be relatively high. In contrast, fewer monitors are observed in parts of Mississippi and West Virginia, where ozone concentrations are generally lower. This visible association between monitor density and ozone level suggests that treating the sampling locations as noninformative may lead to biased spatial inference. Before fitting the spatial model, the geographic coordinates are mapped to a planar coordinate system using the Mercator projection and standardized to lie in $[0,1]^2$ by component-wise range normalization. The preferential sampling model is then fitted using these rescaled spatial coordinates.

\begin{figure}[!t]
\centering
\includegraphics[
    width=0.97\textwidth
]{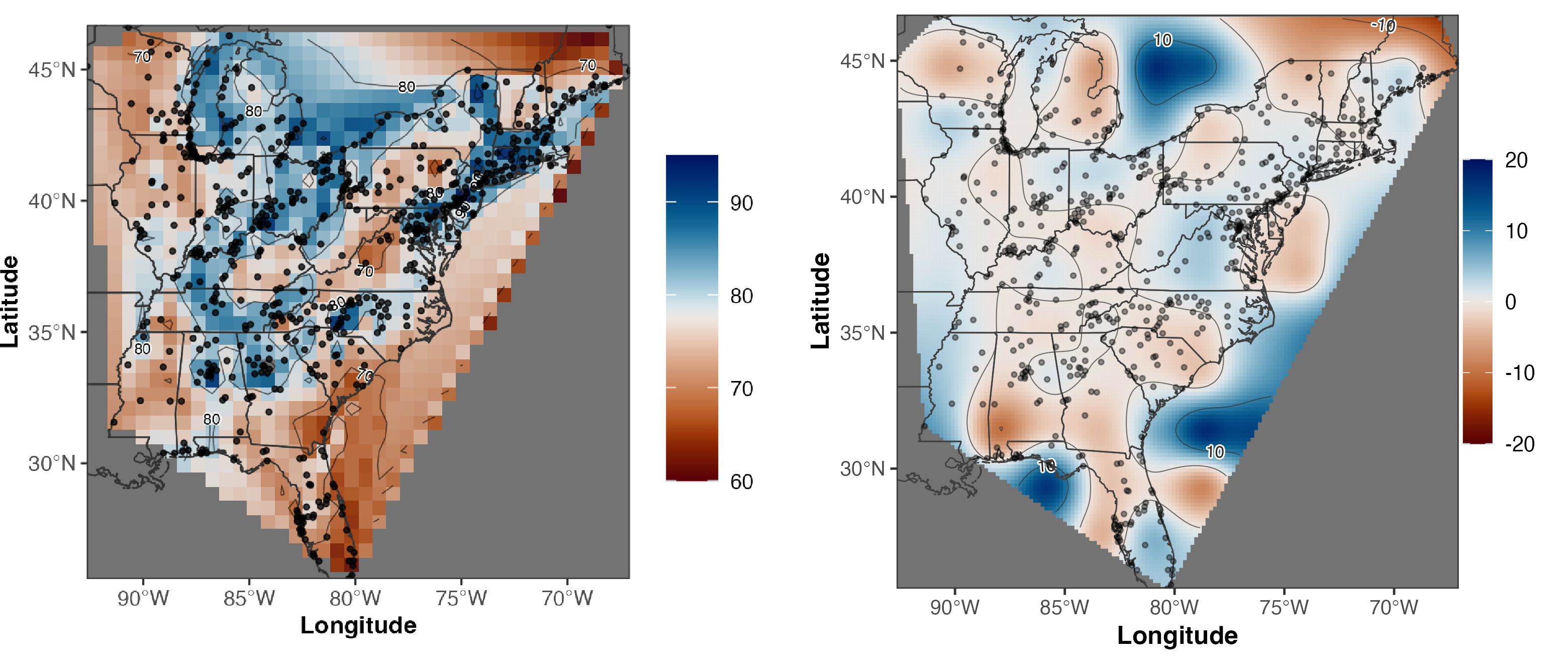}
\caption{Ozone data analysis. Left: spatial distribution of ozone concentrations across the eastern United States, with colors showing ozone levels and black points marking monitoring stations. Right: difference in predicted ozone concentrations between NIS and SME-BETEL, with colors showing predicted values under NIS minus those under SME-BETEL.}
\label{fig:ozone_data_combined}
\end{figure}

We retain the preferential sampling structure in \eqref{full_IS} to account for the possibility that monitoring locations are related to the underlying ozone surface. Let $\mu(\boldsymbol{s})$ denote the spatial surface at location $\boldsymbol{s}\in\mathcal{D}$, where after preprocessing $\mathcal{D}=[0,1]^2$. The observed responses $y_i$ are the ozone measurements collected at monitoring locations $\boldsymbol{s}_i$, $i=1,\ldots,n$. In this real-data analysis, we represent the spatial residual effects using finite basis expansions with zero-mean Gaussian priors on the basis coefficients. We use a reduced kernel-convolution basis to capture the coarse-scale spatial structure shared by the location and response fields, and augment the baseline response field with an intercept and a fine-scale Mat\'ern basis. We assign $\sigma^2\sim\mathrm{Inv\text{-}Ga}(5,100)$ and, for SME-BETEL, additionally assign $a\sim\mathrm{N}(0,10^2)$. For comparison, we consider a noninformative sampling (NIS) model obtained by fixing $a=0$. SME-BETEL applies the mixed-domain construction in \eqref{sm_mixed_type} with the boundary weight in \eqref{mollifier} with $\epsilon=0.2$. For each method, we run four chains of 200{,}000 iterations and discard the first 50{,}000 per chain. Posterior summaries for selected parameters are reported in Table~\ref{tab:ozone_mean_ci}.






\begin{table}[!t]
\centering
\caption{Posterior means and 95\% credible intervals for the ozone data analysis.}
\label{tab:ozone_mean_ci}
\begingroup
\footnotesize
\setlength{\tabcolsep}{4pt}
\renewcommand{\arraystretch}{0.92}
\begin{tabular}{@{}lcc@{}}
\toprule
Parameter & NIS & SME-BETEL \\
\midrule
$a$      & 0 (fixed)        & 3.48 (0.41, 6.86) \\
$\sigma$ & 5.23 (4.96, 5.51) & 5.53 (5.22, 5.85) \\
\bottomrule
\end{tabular}
\endgroup
\end{table}

The 95\% credible interval for $a$ under SME-BETEL excludes zero, indicating evidence of informative sampling in the monitoring network. The right plot of Figure~\ref{fig:ozone_data_combined} shows spatially heterogeneous differences between the NIS and SME-BETEL predictions. The most pronounced discrepancies occur near the boundaries of the spatial domain, particularly around the Great Lakes and along parts of the southern and southeastern coasts. These areas generally have sparser monitoring coverage and are adjacent to regions with relatively high ozone levels, suggesting that accounting for the sampling mechanism can meaningfully affect predictions where local observations are limited but the surrounding ozone surface exhibits substantial spatial variation. Differences are also visible in portions of the Mid-Atlantic and Northeast, where monitoring is considerably denser and observed ozone levels are generally higher. Thus, the effect of accounting for informative sampling is not restricted to sparsely monitored regions: differences between SME-BETEL and NIS can also arise in well-monitored areas where the spatial distribution of monitoring locations may carry information about the underlying ozone surface. Overall, the comparison suggests that incorporating the sampling mechanism changes the estimated ozone surface in a spatially varying manner rather than producing a uniform upward or downward shift.

\section{Conclusion}\label{sec-conc}

We propose SME-BETEL, a semiparametric Bayesian framework that combines score matching estimating equations with Bayesian ETEL for inference in doubly-intractable models, avoiding evaluation of normalizing constants while providing calibrated posterior uncertainty. Building on this framework, we develop a mixed-domain score matching objective that extends SME-BETEL to data whose components lie on different supports. We establish consistency and asymptotic normality of the score matching estimator and a BvM theorem for the SME-BETEL posterior, showing that its limiting covariance matches the sampling covariance of the score matching estimator. Numerical studies demonstrate competitive estimation accuracy under correct specification and reliable uncertainty quantification under model misspecification. The ozone application further illustrates the practical utility of the mixed-domain construction in a preferential sampling setting.

Several directions remain for future work. One is to investigate approaches for improving the statistical efficiency of score matching-based inference while retaining its computational advantages for models with intractable normalizing constants. Another is to develop more systematic constructions of score matching objectives and moment functions for complex or mixed-domain supports. Overall, SME-BETEL provides a practical and theoretically justified route to robust Bayesian inference for unnormalized statistical models.

\section{Supplemental Material}\label{supplemental-material}
The Supplemental Material contains the derivation of mixed-domain score matching and proofs of the theoretical results, additional simulation studies, model-specific formulas and computational details, and the Metropolils--Hastings algorithm used to sample from the SME-BETEL posterior.

\section{Disclosure statement}\label{disclosure-statement}
The authors report there are no competing interests to declare.

\section{Data Availability Statement}\label{data-availability-statement}
The data and code supporting the findings of this study are available in an anonymized repository at \url{https://anonymous.4open.science/r/SME_BETEL-5B4C/}.















\bibliography{bibliography}


\clearpage
\setcounter{page}{1}

\appendix
\begin{center}
    \large\bf Supplementary Material for ``Robust Bayesian Inference for Unnormalized Models with Mixed-Domain Data''
\end{center}
\setcounter{section}{0}
\setcounter{assumption}{0}
\setcounter{theorem}{0}
\setcounter{equation}{0}
\setcounter{algorithm}{0}
\setcounter{table}{0}
\setcounter{figure}{0}

\renewcommand{\theequation}{S\arabic{equation}}
\renewcommand{\thetable}{S\arabic{table}}
\renewcommand{\thefigure}{S\arabic{figure}}
\renewcommand{\thealgorithm}{S\arabic{algorithm}}
\renewcommand{\theassumption}{S\arabic{assumption}}
\renewcommand{\thetheorem}{S\arabic{theorem}}
\renewcommand\thesection{S\arabic{section}}
\renewcommand\thesubsection{S\arabic{section}.\arabic{subsection}}
\makeatletter
\@ifundefined{theHassumption}{}{\renewcommand{\theHassumption}{S\arabic{assumption}}}
\@ifundefined{theHtheorem}{}{\renewcommand{\theHtheorem}{S\arabic{theorem}}}
\@ifundefined{theHtable}{}{\renewcommand{\theHtable}{S\arabic{table}}}
\@ifundefined{theHfigure}{}{\renewcommand{\theHfigure}{S\arabic{figure}}}
\@ifundefined{theHalgorithm}{}{\renewcommand{\theHalgorithm}{S\arabic{algorithm}}}
\makeatother

The supplementary materials are organized as follows. Section~\ref{app:prove} provides the theoretical proofs, including the derivation of mixed-domain score matching, consistency and asymptotic normality of the score matching estimator, and the Bernstein–von Mises theorem for the SME-BETEL posterior. Section~\ref{app:simulation_res} presents supplementary simulation studies and numerical details. Section~\ref{app:formulas_computation_details} collects model-specific formulas and computational details. Section~\ref{app:sampling_algo} presents the Metropolis–Hastings algorithm used to sample from the SME-BETEL posterior.

\section{Proofs of Main Results}
\label{app:prove}

\subsection{Proof of Proposition~\ref{prop:mixed_domain_sm}}
\label{app:heterogeneous_SM_derivation}

In this subsection, we derive the score matching objective \eqref{sm_mixed_type} used for mixed-domain observations
$\mathbf{x} = (\mathbf{y},\mathbf{s})$, where $\mathbf{y}\in \mathbb{R}^d$ and $\mathbf{s}\in \mathcal{D}$.
The main difference from the usual score matching derivation is that the two components
have different supports. The response component $\mathbf{y}$ lies in an unconstrained
Euclidean space, while the location component $\mathbf{s}$ lies in a compact domain.
To handle the boundary of $\mathcal{D}$, we introduce a smooth weight function $h(\mathbf{s})$
that vanishes on $\partial \mathcal{D}$. Write
\[
p_{\mathbf{x}}(\mathbf{y},\mathbf{s})
=
p_{\mathbf{x}}(\mathbf{y}\mid \mathbf{s})p_{\mathbf{x}}(\mathbf{s}),
\qquad
p(\mathbf{y},\mathbf{s};\boldsymbol{\theta})
=
p(\mathbf{y}\mid \mathbf{s};\boldsymbol{\theta})
p(\mathbf{s};\boldsymbol{\theta}).
\]
Define the model and data score functions by
\[
\psi(\mathbf{y},\mathbf{s};\boldsymbol{\theta})
=
\nabla_{(\mathbf{y},\mathbf{s})}
\log p(\mathbf{y},\mathbf{s};\boldsymbol{\theta}),
\qquad
\psi_{\mathbf{x}}(\mathbf{y},\mathbf{s})
=
\nabla_{(\mathbf{y},\mathbf{s})}
\log p_{\mathbf{x}}(\mathbf{y},\mathbf{s}).
\] 
We start from the ``weighted'' relative-Fisher divergence,
\begin{align*}
J(\boldsymbol{\theta})
&=
\frac{1}{2}
\int_{\mathbb{R}^d}\int_{\mathcal{D}}
h(\mathbf{s})
\left\|
\psi(\mathbf{y},\mathbf{s};\boldsymbol{\theta})
-
\psi_{\mathbf{x}}(\mathbf{y},\mathbf{s})
\right\|^2
p_{\mathbf{x}}(\mathbf{y},\mathbf{s})
\,\mu(d\mathbf{s})\,\nu(d\mathbf{y}) \\
&\overset{c}{=}
\frac{1}{2}
\int_{\mathbb{R}^d}\int_{\mathcal{D}}
h(\mathbf{s})
\left\|
\psi(\mathbf{y},\mathbf{s};\boldsymbol{\theta})
\right\|^2
p_{\mathbf{x}}(\mathbf{y},\mathbf{s})
\,\mu(d\mathbf{s})\,\nu(d\mathbf{y}) \\
&\quad -
\int_{\mathbb{R}^d}\int_{\mathcal{D}}
h(\mathbf{s})
\psi(\mathbf{y},\mathbf{s};\boldsymbol{\theta})^{\top}
\psi_{\mathbf{x}}(\mathbf{y},\mathbf{s})
p_{\mathbf{x}}(\mathbf{y},\mathbf{s})
\,\mu(d\mathbf{s})\,\nu(d\mathbf{y}).
\end{align*}
Here, $\overset{c}{=}$ denotes equality up to an additive constant independent of $\boldsymbol{\theta}$. Therefore, the only term that still involves the unknown data-generating density score is the cross term. We next rewrite this cross-term by separating the gradients
with respect to $\mathbf{y}$ and $\mathbf{s}$. We have
\begin{align*}
J(\boldsymbol{\theta})
&\overset{c}{=}
\frac{1}{2}
\int_{\mathbb{R}^d}\int_{\mathcal{D}}
h(\mathbf{s})
\left\|
\psi(\mathbf{y},\mathbf{s};\boldsymbol{\theta})
\right\|^2
p_{\mathbf{x}}(\mathbf{y},\mathbf{s})
\,\mu(d\mathbf{s})\,\nu(d\mathbf{y}) \\
&\quad -
\int_{\mathbb{R}^d}\int_{\mathcal{D}}
h(\mathbf{s})
\left[
\nabla_{(\mathbf{y},\mathbf{s})}
\log p(\mathbf{y},\mathbf{s};\boldsymbol{\theta})
\right]^{\top}
\left[
\nabla_{(\mathbf{y},\mathbf{s})}
\log p_{\mathbf{x}}(\mathbf{y},\mathbf{s})
\right]
p_{\mathbf{x}}(\mathbf{y},\mathbf{s})
\,\mu(d\mathbf{s})\,\nu(d\mathbf{y}) \\
&=
\frac{1}{2}
\int_{\mathbb{R}^d}\int_{\mathcal{D}}
h(\mathbf{s})
\left\|
\psi(\mathbf{y},\mathbf{s};\boldsymbol{\theta})
\right\|^2
p_{\mathbf{x}}(\mathbf{y},\mathbf{s})
\,\mu(d\mathbf{s})\,\nu(d\mathbf{y}) \\
&\quad -
\int_{\mathbb{R}^d}\int_{\mathcal{D}}
h(\mathbf{s})
\left[
\nabla_{\mathbf{y}}
\log p(\mathbf{y}\mid \mathbf{s};\boldsymbol{\theta})
\right]^{\top}
\nabla_{\mathbf{y}}
\log p_{\mathbf{x}}(\mathbf{y}\mid \mathbf{s})
p_{\mathbf{x}}(\mathbf{y},\mathbf{s})
\,\mu(d\mathbf{s})\,\nu(d\mathbf{y}) \\
&\quad -
\int_{\mathbb{R}^d}\int_{\mathcal{D}}
h(\mathbf{s})
\left[
\nabla_{\mathbf{s}}
\log p(\mathbf{y},\mathbf{s};\boldsymbol{\theta})
\right]^{\top}
\nabla_{\mathbf{s}}
\log p_{\mathbf{x}}(\mathbf{y},\mathbf{s})
p_{\mathbf{x}}(\mathbf{y},\mathbf{s})
\,\mu(d\mathbf{s})\,\nu(d\mathbf{y}).
\end{align*}
The last two terms are handled by integration by parts. For the $\mathbf{y}$-term, the
usual Euclidean score matching argument applies. For the $\mathbf{s}$-term, integration by
parts is applied over the compact domain $D$; the boundary contribution vanishes because
$h(\mathbf{s})=0$ on $\partial \mathcal{D}$. This gives
\begin{align*}
J(\boldsymbol{\theta})
&\overset{c}{=}
\frac{1}{2}
\int_{\mathbb{R}^d}\int_{\mathcal{D}}
h(\mathbf{s})
\left\|
\psi(\mathbf{y},\mathbf{s};\boldsymbol{\theta})
\right\|^2
p_{\mathbf{x}}(\mathbf{y},\mathbf{s})
\,\mu(d\mathbf{s})\,\nu(d\mathbf{y}) \\
&\quad +
\int_{\mathbb{R}^d}\int_{\mathcal{D}}
h(\mathbf{s})
\operatorname{tr}
\left\{
\nabla_{\mathbf{y}}^2
\log p(\mathbf{y}\mid \mathbf{s};\boldsymbol{\theta})
\right\}
p_{\mathbf{x}}(\mathbf{y},\mathbf{s})
\,\mu(d\mathbf{s})\,\nu(d\mathbf{y}) \\
&\quad +
\int_{\mathbb{R}^d}\int_{\mathcal{D}}
\bigg[
h(\mathbf{s})
\Delta_{\mathbf{s}}
\log p(\mathbf{y},\mathbf{s};\boldsymbol{\theta})
+
\nabla_{\mathbf{s}}h(\mathbf{s})^{\top}
\nabla_{\mathbf{s}}
\log p(\mathbf{y},\mathbf{s};\boldsymbol{\theta})
\bigg]
p_{\mathbf{x}}(\mathbf{y},\mathbf{s})
\,\mu(d\mathbf{s})\,\nu(d\mathbf{y}) \\
&=
\int_{\mathbb{R}^d}\int_{\mathcal{D}}
\bigg[
\frac{1}{2}
h(\mathbf{s})
\left\|
\psi(\mathbf{y},\mathbf{s};\boldsymbol{\theta})
\right\|^2
+
h(\mathbf{s})
\operatorname{tr}
\left\{
\nabla_{\mathbf{y}}^2
\log p(\mathbf{y}\mid \mathbf{s};\boldsymbol{\theta})
\right\} \\
&\qquad\qquad
+
h(\mathbf{s})
\Delta_{\mathbf{s}}
\log p(\mathbf{y},\mathbf{s};\boldsymbol{\theta})
+
\nabla_{\mathbf{s}}h(\mathbf{s})^{\top}
\nabla_{\mathbf{s}}
\log p(\mathbf{y},\mathbf{s};\boldsymbol{\theta})
\bigg]
p_{\mathbf{x}}(\mathbf{y},\mathbf{s})
\,\mu(d\mathbf{s})\,\nu(d\mathbf{y}) \\
&=
\mathbb{E}_{p_{\mathbf{x}}}
\bigg[
\frac{1}{2}
h(\mathbf{s})
\left\|
\psi(\mathbf{y},\mathbf{s};\boldsymbol{\theta})
\right\|^2
+
h(\mathbf{s})
\operatorname{tr}
\left\{
\nabla_{\mathbf{y}}^2
\log p(\mathbf{y}\mid \mathbf{s};\boldsymbol{\theta})
\right\} \\
&\qquad\qquad
+
h(\mathbf{s})
\Delta_{\mathbf{s}}
\log p(\mathbf{y},\mathbf{s};\boldsymbol{\theta})
+
\nabla_{\mathbf{s}}h(\mathbf{s})^{\top}
\nabla_{\mathbf{s}}
\log p(\mathbf{y},\mathbf{s};\boldsymbol{\theta})
\bigg].
\end{align*}
This is the tractable mixed-domain score matching objective. It depends only on derivatives
of the working model and does not involve the unknown score
$\nabla_{(\mathbf{y},\mathbf{s})}\log p_{\mathbf{x}}(\mathbf{y},\mathbf{s})$.

\subsection{Consistency of Sample Score Matching Estimator}
\label{app:sm_consistency}
In this subsection, we prove the consistency of the sample score matching estimator $\widehat{\boldsymbol{\theta}}_n$ obtained from \eqref{sm_sampleversion}. We begin with regularity conditions following \citet{hyvarinen2005estimation}:
\begin{assumption}
\label{a_1}
    The population score matching objective function $J(\boldsymbol{\theta})$ has a unique minimizer over $\Theta$; that is, for every $\boldsymbol{\theta}\in\Theta$ with $\boldsymbol{\theta}\neq \boldsymbol{\theta}^*$, $J(\boldsymbol{\theta}) > J(\boldsymbol{\theta}^*)$.
\end{assumption}
\begin{assumption}
\label{a_2}
    $\mathbb{E}_{p_{\mathbf{x}}}\{\|\psi(\mathbf{x}; \boldsymbol{\theta})\|^2\}$ and $\mathbb{E}_{p_{\mathbf{x}}}\{\|\psi_{\mathbf{x}}(\mathbf{x})\|^2\}$ are finite for any $\boldsymbol{\theta} \in \Theta$.
\end{assumption}
\begin{assumption}
\label{a_3}
    For any $\boldsymbol{\theta} \in \Theta$, $\lim_{\|\mathbf{x}\| \rightarrow \infty}p_{\mathbf{x}}(\mathbf{x})\psi(\mathbf{x};\boldsymbol{\theta}) = \mathbf{0}$.
\end{assumption}
\begin{assumption}
\label{a_4}
    The model density $p(\mathbf{x}; \boldsymbol{\theta}) > 0$ for any $\mathbf{x} \in \mathbb{R}^m$ and $\boldsymbol{\theta} \in \Theta$.
\end{assumption}

Assumption~\ref{a_1} is an identification condition for the score matching target $\boldsymbol{\theta}^*$ defined in \eqref{theta^*}. Assumptions~\ref{a_2} and~\ref{a_3} control the behavior of the model score function and the score function of the data-generating density. In particular, Assumption~\ref{a_3} plays a crucial role in deriving \eqref{sm_2}. Assumption~\ref{a_4} requires that all densities are strictly positive. We further impose the following Lipschitz condition on the model score function:
\begin{assumption}
\label{a_6}
    Both $\nabla_{\mathbf{x}} \psi(\mathbf{x}; \boldsymbol{\theta})$ and $\psi(\mathbf{x}; \boldsymbol{\theta})\psi^\top(\mathbf{x}; \boldsymbol{\theta})$ are Lipschitz continuous in Frobenius norm, i.e., for any $\boldsymbol{\theta}_1$ and $\boldsymbol{\theta}_2 \in \Theta$, $\|\nabla_{\mathbf{x}} \psi(\mathbf{x}; \boldsymbol{\theta}_1) - \nabla_{\mathbf{x}} \psi(\mathbf{x}; \boldsymbol{\theta}_2)\|_F \leq L_1(\mathbf{x})\|\boldsymbol{\theta}_1 - \boldsymbol{\theta}_2 \|_2$ and $\|\psi(\mathbf{x}; \boldsymbol{\theta}_1)\psi^\top(\mathbf{x}; \boldsymbol{\theta}_1)\ - \psi(\mathbf{x}; \boldsymbol{\theta}_2)\psi^\top(\mathbf{x}; \boldsymbol{\theta}_2) \|_F \leq L_2(\mathbf{x})\|\boldsymbol{\theta}_1 - \boldsymbol{\theta}_2 \|_2$. Furthermore, we require $\mathbb{E}_{p_{\mathbf{x}}}\{L_1^2(\mathbf{x})\} < \infty$ and $\mathbb{E}_{p_{\mathbf{x}}}\{L_2^2(\mathbf{x})\} < \infty$.
\end{assumption}
Under this additional condition, we establish the consistency of $\widehat{\boldsymbol{\theta}}_n$.

\begin{lemma}
\label{lemma2}
    Suppose Assumption~\ref{a_6} holds. Let $f(\mathbf{x};\boldsymbol{\theta}) = \operatorname{tr}\{\nabla_{\mathbf{x}}\psi(\mathbf{x};\boldsymbol{\theta})\} + 2^{-1}\|\psi(\mathbf{x};\boldsymbol{\theta})\|_2^2.$ Then $f(\mathbf{x};\boldsymbol{\theta})$ is Lipschitz continuous in $\boldsymbol{\theta}$ with Lipschitz constant $L(\mathbf{x})$ satisfying $\mathbb{E}_{p_{\mathbf{x}}}\{L^2(\mathbf{x})\}<\infty.$
\end{lemma}

\begin{proof}
    For notational simplicity, define $A(\boldsymbol{\theta})=\nabla_{\mathbf{x}}\psi(\mathbf{x};\boldsymbol{\theta})$ and $B(\boldsymbol{\theta})=\psi(\mathbf{x};\boldsymbol{\theta})\psi(\mathbf{x};\boldsymbol{\theta})^\top$. Then $\|\psi(\mathbf{x};\boldsymbol{\theta})\|_2^2=\operatorname{tr}\{B(\boldsymbol{\theta})\}$. For any $\boldsymbol{\theta}_1,\boldsymbol{\theta}_2\in\Theta$, we have
\begin{align*}
    &\left|f(\mathbf{x};\boldsymbol{\theta}_1) - f(\mathbf{x};\boldsymbol{\theta}_2) \right|  \\
    &\quad = \left|\operatorname{tr}\{A(\boldsymbol{\theta}_1)-A(\boldsymbol{\theta}_2)\} + \frac12 \operatorname{tr}\{B(\boldsymbol{\theta}_1)-B(\boldsymbol{\theta}_2)\} \right|  \\
    &\quad \leq \left| \operatorname{tr}\{A(\boldsymbol{\theta}_1)-A(\boldsymbol{\theta}_2)\} \right| + \frac12 \left|\operatorname{tr}\{B(\boldsymbol{\theta}_1)-B(\boldsymbol{\theta}_2)\} \right|  \\
    &\quad \leq \sqrt{m} \left\|A(\boldsymbol{\theta}_1)-A(\boldsymbol{\theta}_2)\right\|_F + \frac{\sqrt{m}}{2}\left\|B(\boldsymbol{\theta}_1)-B(\boldsymbol{\theta}_2)\right\|_F  \\
    &\quad \leq \sqrt{m} \left\{L_1(\mathbf{x}) + \frac12L_2(\mathbf{x}) \right\} \|\boldsymbol{\theta}_1-\boldsymbol{\theta}_2\|_2,
\end{align*}
where the second inequality uses $|\operatorname{tr}(M)|\leq \sqrt{m}\|M\|_F$ for any $m\times m$ matrix $M$, and the final inequality follows from Assumption~\ref{a_6}. Define $L(\mathbf{x})=\sqrt{m}\{L_1(\mathbf{x})+L_2(\mathbf{x})/2\}$. Then $f(\mathbf{x};\boldsymbol{\theta})$ is Lipschitz continuous in $\boldsymbol{\theta}$ with Lipschitz constant $L(\mathbf{x})$. Moreover,
\begin{align*}
    \mathbb{E}_{p_{\mathbf{x}}}\{L^2(\mathbf{x})\}
    &=m\,\mathbb{E}_{p_{\mathbf{x}}}\left[\left\{L_1(\mathbf{x})+\frac12 L_2(\mathbf{x})\right\}^2\right]  \\
    &\leq 2m\,\mathbb{E}_{p_{\mathbf{x}}}\{L_1^2(\mathbf{x})\}+\frac{m}{2}\,\mathbb{E}_{p_{\mathbf{x}}}\{L_2^2(\mathbf{x})\} \\
    &< \infty,
\end{align*}
where the final inequality follows from Assumption~\ref{a_6}. This completes the proof.
\end{proof}

\begin{lemma}
\label{lemma3}
    Suppose Assumption~\ref{a_6} holds and $\Theta$ is compact, as imposed in Assumption~\ref{s_2}. In addition, assume that $\mathbb{E}_{p_{\mathbf{x}}}\{f^2(\mathbf{x};\boldsymbol{\theta}^*)\}<\infty$. Then
\begin{align*}
    \mathbb{E}\bigg[\sup_{\boldsymbol{\theta}\in\Theta}\big|\widehat{J}(\boldsymbol{\theta};\mathbf{x}_{1:n})-J(\boldsymbol{\theta})\big|\bigg] \leq C\left\{\frac{1}{\sqrt{n}} + \operatorname{diam}(\Theta)\sqrt{\frac{p}{n}}\right\},
\end{align*}
where $C<\infty$ is a constant independent of $n$, $\operatorname{diam}(\Theta)$ denotes the Euclidean diameter of $\Theta$, and $p$ is fixed and denotes the dimension of $\Theta$.
\end{lemma}

\begin{proof}
    Write $f_{\boldsymbol{\theta}}(\mathbf{x})=f(\mathbf{x};\boldsymbol{\theta})$, $P_nf_{\boldsymbol{\theta}}=n^{-1}\sum_{i=1}^n f_{\boldsymbol{\theta}}(\mathbf{x}_i)$, and $Pf_{\boldsymbol{\theta}}=\mathbb{E}_{p_{\mathbf{x}}}\{f_{\boldsymbol{\theta}}(\mathbf{x})\}$. Since the additive constant in the score matching objective does not depend on $\boldsymbol{\theta}$, it suffices to control $\sup_{\boldsymbol{\theta}\in\Theta}|P_nf_{\boldsymbol{\theta}}-Pf_{\boldsymbol{\theta}}|$. Let $\mathbf{x}_{1:n}'$ be an independent copy of $\mathbf{x}_{1:n}$. By symmetrization,
\begin{align*}
    \mathbb{E}\bigg[\sup_{\boldsymbol{\theta}\in\Theta}|P_nf_{\boldsymbol{\theta}}-Pf_{\boldsymbol{\theta}}|\bigg] &\leq \mathbb{E}\bigg[\sup_{\boldsymbol{\theta}\in\Theta}|P_nf_{\boldsymbol{\theta}}-P_n'f_{\boldsymbol{\theta}}|\bigg]  \\
    &\leq 2\mathbb{E}\bigg[ \sup_{\boldsymbol{\theta}\in\Theta}\bigg|\frac{1}{n}\sum_{i=1}^n \epsilon_i f_{\boldsymbol{\theta}}(\mathbf{x}_i)\bigg|\bigg],
\end{align*}
where $P_n'f_{\boldsymbol{\theta}}=n^{-1}\sum_{i=1}^n f_{\boldsymbol{\theta}}(\mathbf{x}_i')$, and $\epsilon_1,\ldots,\epsilon_n$ are independent Rademacher random variables independent of the data. Define $Z_n(\boldsymbol{\theta})=n^{-1}\sum_{i=1}^n \epsilon_i f_{\boldsymbol{\theta}}(\mathbf{x}_i)$. Then
\begin{align*}
    \mathbb{E}\bigg[\sup_{\boldsymbol{\theta}\in\Theta}|Z_n(\boldsymbol{\theta})|\bigg] &\leq \mathbb{E}\{|Z_n(\boldsymbol{\theta}^*)|\} + \mathbb{E}\bigg[\sup_{\boldsymbol{\theta}\in\Theta}|Z_n(\boldsymbol{\theta})-Z_n(\boldsymbol{\theta}^*)| \bigg].
\end{align*}
For the first term, by Cauchy--Schwarz and the assumption $\mathbb{E}_{p_{\mathbf{x}}}\{f^2(\mathbf{x};\boldsymbol{\theta}^*)\}<\infty$,
\begin{align*}
    \mathbb{E}\{|Z_n(\boldsymbol{\theta}^*)|\} &\leq \left[\mathbb{E}\{Z_n^2(\boldsymbol{\theta}^*)\}\right]^{1/2} = \left[\frac{1}{n}\mathbb{E}_{p_{\mathbf{x}}}\{f^2(\mathbf{x};\boldsymbol{\theta}^*)\}\right]^{1/2} = O(n^{-1/2}).
\end{align*}
It remains to bound the increment process. Conditional on $\mathbf{x}_{1:n}$, for any $\boldsymbol{\theta}_1,\boldsymbol{\theta}_2\in\Theta$, Lemma~\ref{lemma2} gives
\begin{align*}
    \mathbb{E}_{\epsilon}\exp\left[\lambda\{Z_n(\boldsymbol{\theta}_1)-Z_n(\boldsymbol{\theta}_2)\}\right] &\leq \exp\left\{\frac{\lambda^2}{2n^2}\sum_{i=1}^n\left[f_{\boldsymbol{\theta}_1}(\mathbf{x}_i) - f_{\boldsymbol{\theta}_2}(\mathbf{x}_i)\right]^2\right\}  \\
    &\leq \exp\left\{\frac{\lambda^2}{2}d_n^2(\boldsymbol{\theta}_1,\boldsymbol{\theta}_2)\right\},
\end{align*}
where $d_n(\boldsymbol{\theta}_1,\boldsymbol{\theta}_2)=n^{-1/2}\{n^{-1}\sum_{i=1}^n L^2(\mathbf{x}_i)\}^{1/2}\|\boldsymbol{\theta}_1-\boldsymbol{\theta}_2\|_2$. Hence $Z_n(\boldsymbol{\theta})-Z_n(\boldsymbol{\theta}^*)$ is a sub-Gaussian process with respect to $d_n$. For $u>0$, let $N(\Theta,d_n,u)$ denote the $u$-covering number of $\Theta$ under the metric $d_n$; that is, $N(\Theta,d_n,u)$ is the smallest integer $N$ such that there exist points $\boldsymbol{\theta}_1,\ldots,\boldsymbol{\theta}_N\in\Theta$ satisfying $\Theta \subset \bigcup_{k=1}^N \left\{\boldsymbol{\theta}\in\Theta: d_n(\boldsymbol{\theta},\boldsymbol{\theta}_k)\leq u \right\}.$ By Dudley's entropy integral \citep{dudley1967sizes},
\begin{align*}
    \mathbb{E}_{\epsilon}\bigg[\sup_{\boldsymbol{\theta}\in\Theta}|Z_n(\boldsymbol{\theta})-Z_n(\boldsymbol{\theta}^*)|\,\bigg|\,\mathbf{x}_{1:n}\bigg] &\leq C\int_0^{D_n} \sqrt{\log N(\Theta,d_n,u)}\,du,
\end{align*}
where $D_n=n^{-1/2}\{n^{-1}\sum_{i=1}^n L^2(\mathbf{x}_i)\}^{1/2} \operatorname{diam}(\Theta)$. Since $\Theta\subset\mathbb{R}^p$ is compact, $N(\Theta,d_n,u)\leq (1+C D_n/u)^p$ for a universal constant $C>0$. Therefore,
\begin{align*}
    \int_0^{D_n} \sqrt{\log N(\Theta,d_n,u)}\,du &\leq C D_n\sqrt{p}.
\end{align*}
Taking expectation over $\mathbf{x}_{1:n}$ and using Jensen's inequality,
\begin{align*}
    \mathbb{E}\bigg[\sup_{\boldsymbol{\theta}\in\Theta}|Z_n(\boldsymbol{\theta})-Z_n(\boldsymbol{\theta}^*)|\bigg] &\leq C\operatorname{diam}(\Theta)\sqrt{\frac{p}{n}}\, \mathbb{E}\left[\left\{\frac{1}{n}\sum_{i=1}^n L^2(\mathbf{x}_i)\right\}^{1/2}\right]  \\
    &\leq C\operatorname{diam}(\Theta)\sqrt{\frac{p}{n}}\,\left[\mathbb{E}_{p_{\mathbf{x}}}\{L^2(\mathbf{x})\}\right]^{1/2}  \\
    &= O\left(\operatorname{diam}(\Theta)\sqrt{\frac{p}{n}}\right),
\end{align*}
where the final equality follows from Lemma~\ref{lemma2}. Combining the preceding bounds yields
\begin{align*}
    \mathbb{E}\bigg[\sup_{\boldsymbol{\theta}\in\Theta}\big|\widehat{J}(\boldsymbol{\theta};\mathbf{x}_{1:n}) - J(\boldsymbol{\theta})\big|\bigg] &\leq C\left\{\frac{1}{\sqrt{n}} + \operatorname{diam}(\Theta)\sqrt{\frac{p}{n}}\right\}.
\end{align*}
This completes the proof.
\end{proof}

\begin{theorem}
\label{consistency}
    (Consistency) Suppose Assumptions~\ref{a_1}--\ref{a_6} hold, and suppose the observations are i.i.d. with common density $p_{\mathbf{x}}$, as imposed in Assumption~\ref{s_3}(a). Suppose also that $\Theta$ is compact, as imposed in Assumption~\ref{s_2}. In addition, assume that $\mathbb{E}_{p_{\mathbf{x}}}\{f^2(\mathbf{x};\boldsymbol{\theta}^*)\}<\infty$. Let $\widehat{\boldsymbol{\theta}}_n$ be any measurable sample score matching estimator, that is,
    \begin{align*}
        \widehat{\boldsymbol{\theta}}_n \in \argmin_{\boldsymbol{\theta}\in\Theta} \widehat{J}(\boldsymbol{\theta};\mathbf{x}_{1:n}).
    \end{align*}
    Then, as $n\to\infty$,
    \begin{align*}
        \widehat{\boldsymbol{\theta}}_n \convp \boldsymbol{\theta}^*,
    \end{align*}
    where $\boldsymbol{\theta}^*$ is the unique minimizer of the population score matching objective $J(\boldsymbol{\theta})$ over $\Theta$.
\end{theorem}

\begin{proof}
    For notational simplicity, write $\widehat{J}_n(\boldsymbol{\theta})=\widehat{J}(\boldsymbol{\theta};\mathbf{x}_{1:n})$. Since $\boldsymbol{\theta}^*$ minimizes $J$ and $\widehat{\boldsymbol{\theta}}_n$ minimizes $\widehat{J}_n$, we have
\begin{align*}
    0 \leq J(\widehat{\boldsymbol{\theta}}_n)-J(\boldsymbol{\theta}^*) &= \{J(\widehat{\boldsymbol{\theta}}_n)-\widehat{J}_n(\widehat{\boldsymbol{\theta}}_n)\} + \{\widehat{J}_n(\widehat{\boldsymbol{\theta}}_n)-\widehat{J}_n(\boldsymbol{\theta}^*)\} + \{\widehat{J}_n(\boldsymbol{\theta}^*)-J(\boldsymbol{\theta}^*)\}  \\
    &\leq 2\sup_{\boldsymbol{\theta}\in\Theta}\left|\widehat{J}_n(\boldsymbol{\theta})-J(\boldsymbol{\theta})\right| .
\end{align*}
Therefore, for every $\eta>0$, Lemma~\ref{lemma3} and Markov's inequality imply
\begin{align*}
    P\left(J(\widehat{\boldsymbol{\theta}}_n)-J(\boldsymbol{\theta}^*)>\eta \right) &\leq P\left(2\sup_{\boldsymbol{\theta}\in\Theta}\left|\widehat{J}_n(\boldsymbol{\theta})-J(\boldsymbol{\theta})\right|>\eta \right)  \\
    &\leq \frac{2}{\eta} \mathbb{E}\left[\sup_{\boldsymbol{\theta}\in\Theta}\left|\widehat{J}_n(\boldsymbol{\theta})-J(\boldsymbol{\theta})\right|\right] \\
    &\to 0 .
\end{align*}
Hence $J(\widehat{\boldsymbol{\theta}}_n) \convp J(\boldsymbol{\theta}^*)$. It remains to show that convergence of the objective values implies convergence of the estimators. Recall $f(\mathbf{x};\boldsymbol{\theta})=\operatorname{tr}\{\nabla_{\mathbf{x}}\psi(\mathbf{x};\boldsymbol{\theta})\} + 2^{-1}\|\psi(\mathbf{x};\boldsymbol{\theta})\|_2^2$. Since the additive constant in $J(\boldsymbol{\theta})$ does not depend on $\boldsymbol{\theta}$, Lemma~\ref{lemma2} gives, for any $\boldsymbol{\theta}_1,\boldsymbol{\theta}_2\in\Theta$,
\begin{align*}
    \left|J(\boldsymbol{\theta}_1)-J(\boldsymbol{\theta}_2)\right| &\leq \mathbb{E}_{p_{\mathbf{x}}}\left[\left|f(\mathbf{x};\boldsymbol{\theta}_1) - f(\mathbf{x};\boldsymbol{\theta}_2)\right|\right]  \\
    &\leq \mathbb{E}_{p_{\mathbf{x}}}\{L(\mathbf{x})\} \|\boldsymbol{\theta}_1-\boldsymbol{\theta}_2\|_2 \\
    &<\infty .
\end{align*}
It follows that $J$ is Lipschitz continuous on $\Theta$. Indeed, since $\mathbb{E}_{p_{\mathbf{x}}}\{L(\mathbf{x})\} \leq \{\mathbb{E}_{p_{\mathbf{x}}}\{L^2(\mathbf{x})\}\}^{1/2}<\infty$, the preceding inequality gives
\begin{align*}
    \left|J(\boldsymbol{\theta}_1)-J(\boldsymbol{\theta}_2)\right| &\leq C_J \|\boldsymbol{\theta}_1-\boldsymbol{\theta}_2\|_2,
\end{align*}
for every $\boldsymbol{\theta}_1,\boldsymbol{\theta}_2\in\Theta$, where $C_J=\mathbb{E}_{p_{\mathbf{x}}}\{L(\mathbf{x})\}<\infty$. Hence $J$ is continuous on $\Theta$. Fix $\epsilon>0$ and define $\mathcal{S}_{\epsilon}=\{\boldsymbol{\theta}\in\Theta: \|\boldsymbol{\theta}-\boldsymbol{\theta}^*\|_2\geq \epsilon\}$. If $\mathcal{S}_{\epsilon}=\varnothing$, then $P(\|\widehat{\boldsymbol{\theta}}_n-\boldsymbol{\theta}^*\|_2\geq\epsilon)=0$. Otherwise, $\mathcal{S}_{\epsilon}$ is compact. Since $J$ is continuous and $\boldsymbol{\theta}^*$ is the unique minimizer of $J$, we have
\begin{align*}
    c_{\epsilon} := \inf_{\boldsymbol{\theta}\in\mathcal{S}_{\epsilon}} \left\{J(\boldsymbol{\theta})-J(\boldsymbol{\theta}^*)\right\} >0 .
\end{align*}
Consequently,
\begin{align*}
    P\left(\|\widehat{\boldsymbol{\theta}}_n-\boldsymbol{\theta}^*\|_2\geq\epsilon \right) &\leq P\left(J(\widehat{\boldsymbol{\theta}}_n)-J(\boldsymbol{\theta}^*)\geq c_{\epsilon}\right) \to 0 .
\end{align*}
Since $\epsilon>0$ is arbitrary, $\widehat{\boldsymbol{\theta}}_n\convp\boldsymbol{\theta}^*$.
\end{proof}

\subsection{Asymptotic normality}
\label{app:sm_asymptotic}
\begin{theorem}
\label{asymptotic_normality}
    (Asymptotic normality) Suppose the conditions of Theorem~\ref{consistency} hold. In addition, suppose main-text Assumptions~\ref{s_2}--\ref{s_4} hold. Let $\mathbf{g}(\mathbf{x},\boldsymbol{\theta}) = \nabla_{\boldsymbol{\theta}}[\operatorname{tr}\{\nabla_{\mathbf{x}}\psi(\mathbf{x};\boldsymbol{\theta})\} +2^{-1}\|\psi(\mathbf{x};\boldsymbol{\theta})\|_2^2]$. Then, as $n\to\infty$,
    \begin{align*}
    \sqrt{n}(\widehat{\boldsymbol{\theta}}_n-\boldsymbol{\theta}^*) \convd N\left(\mathbf{0}, (\Gamma^\top\Delta^{-1}\Gamma)^{-1}\right),
    \end{align*}
where $\Delta$ and $\Gamma$ are defined in Assumption~\ref{s_4}.
\end{theorem}

\begin{proof}
    Write $\widehat{J}_n(\boldsymbol{\theta})=\widehat{J}(\boldsymbol{\theta};\mathbf{x}_{1:n})$ and $G(\mathbf{x},\boldsymbol{\theta})=\partial\mathbf{g}(\mathbf{x},\boldsymbol{\theta})/\partial\boldsymbol{\theta}^{\top}$. Since $\boldsymbol{\theta}^*\in\operatorname{int}(\Theta)$ by Assumption~\ref{s_2} and $\widehat{\boldsymbol{\theta}}_n\convp\boldsymbol{\theta}^*$ by Theorem~\ref{consistency}, we have $P\big(\widehat{\boldsymbol{\theta}}_n\in\operatorname{int}(\Theta)\big)\to 1$. Therefore, with probability tending to one, the empirical first-order condition gives $n^{-1}\sum_{i=1}^n\mathbf{g}(\mathbf{x}_i,\widehat{\boldsymbol{\theta}}_n)=\mathbf{0}$. Moreover, by Assumption~\ref{s_2}, $\mathbb{E}_{p_{\mathbf{x}}}\{\mathbf{g}(\mathbf{x},\boldsymbol{\theta}^*)\}=\mathbf{0}$. Using the integral form of Taylor's expansion around $\boldsymbol{\theta}^*$,
    \begin{align*}
    \mathbf{0} &= \frac{1}{n}\sum_{i=1}^n \mathbf{g}(\mathbf{x}_i,\boldsymbol{\theta}^*) + \widehat{\Gamma}_n(\widehat{\boldsymbol{\theta}}_n-\boldsymbol{\theta}^*),
    \end{align*}
    where $\widehat{\Gamma}_n=n^{-1}\sum_{i=1}^n\int_0^1G\{\mathbf{x}_i,\boldsymbol{\theta}^*+t(\widehat{\boldsymbol{\theta}}_n-\boldsymbol{\theta}^*)\}\,dt$. By Assumption~\ref{s_3}, the uniform law of large numbers in a neighborhood of $\boldsymbol{\theta}^*$, together with $\widehat{\boldsymbol{\theta}}_n\convp\boldsymbol{\theta}^*$, gives $\widehat{\Gamma}_n\convp\Gamma$. Since $\operatorname{rank}(\Gamma)=p$, we also have $\widehat{\Gamma}_n^{-1}\convp\Gamma^{-1}$. Therefore,
    \begin{align*}
    \sqrt{n}(\widehat{\boldsymbol{\theta}}_n-\boldsymbol{\theta}^*) &= -\widehat{\Gamma}_n^{-1}\left\{\frac{1}{\sqrt{n}}\sum_{i=1}^n\mathbf{g}(\mathbf{x}_i,\boldsymbol{\theta}^*)\right\}  \\
    &= -\Gamma^{-1}\left\{\frac{1}{\sqrt{n}}\sum_{i=1}^n\mathbf{g}(\mathbf{x}_i,\boldsymbol{\theta}^*)\right\} + o_p(1).
    \end{align*}
    By Assumption~\ref{s_2}, $\mathbb{E}_{p_{\mathbf{x}}}\{\mathbf{g}(\mathbf{x},\boldsymbol{\theta}^*)\} = \mathbf{0}$. Furthermore, Assumption~\ref{s_3}(a) ensures that $\mathbf{x}_1,\ldots,\mathbf{x}_n$ are i.i.d., and Assumption~\ref{s_4}(a) ensures that $\Delta = \mathbb{E}_{p_{\mathbf{x}}}\{\mathbf{g}(\mathbf{x},\boldsymbol{\theta}^*)\mathbf{g}(\mathbf{x},\boldsymbol{\theta}^*)^\top\}$ is nonsingular. Therefore, the multivariate central limit theorem gives
    \begin{align*}
    \frac{1}{\sqrt{n}}\sum_{i=1}^n\mathbf{g}(\mathbf{x}_i,\boldsymbol{\theta}^*) \convd N(\mathbf{0},\Delta).
    \end{align*}
    Slutsky's theorem then implies
    \begin{align*}
    \sqrt{n} (\widehat{\boldsymbol{\theta}}_n-\boldsymbol{\theta}^*) \convd N\left(\mathbf{0}, \boldsymbol{\Gamma}^{-1}\Delta(\boldsymbol{\Gamma}^{-1})^\top\right).
    \end{align*}
Finally, since $\Gamma$ and $\Delta$ are nonsingular, $\boldsymbol{\Gamma}^{-1}\Delta(\boldsymbol{\Gamma}^{-1})^\top = (\boldsymbol{\Gamma}^{\top}\Delta^{-1}\boldsymbol{\Gamma})^{-1}$. This completes the proof.
\end{proof}

\begin{remark}
    The matrix $(\Gamma^{\top}\Delta^{-1}\Gamma)^{-1}$ is sometimes called inverse Godambe information \citep{godambe1991estimating} and it is a general asymptotic covariance form for solution of unbiased estimating equations.
\end{remark}

\subsection{Proof of Theorem~\ref{bvm}}
\label{app:bvm}

The proof follows the general Bernstein–von Mises argument for Bayesian ETEL posterior developed by \citet{chib2018bayesian}, adapted to the present SME-BETEL setting. Their formulation allows for an augmented parameter vector defined on a product space, which is useful for handling potentially invalid moment restrictions. In contrast, the SME-BETEL posterior \eqref{SME-BETEL_posterior} considered here is defined directly on the parameter space $\Theta$, since the moment function is generated from the first-order condition of the score matching objective. We therefore specialize the argument to this single parameter space setting and write the proof in the notation of the present paper.

Throughout this subsection, define $\boldsymbol{V}:=\Gamma^{\top}\Delta^{-1}\Gamma$ and $\mathbf{h}=\sqrt{n}(\boldsymbol{\theta}-\widehat{\boldsymbol{\theta}}_n)$. We also write
\[
M_n(\boldsymbol{\theta})
=
\sum_{i=1}^n
\ell_{n,\boldsymbol{\theta}}(\mathbf{x}_i),
\qquad
S_n(\boldsymbol{\theta})
=
\frac{\partial M_n(\boldsymbol{\theta})}
{\partial\boldsymbol{\theta}},
\qquad
H_n(\boldsymbol{\theta})
=
\frac{\partial^2 M_n(\boldsymbol{\theta})}
{\partial\boldsymbol{\theta}\partial\boldsymbol{\theta}^{\top}},
\]
and denote the local ETEL likelihood ratio by
\[
R_n(\mathbf{h})
=
\frac{
L(\mathbf{x}_{1:n}\mid \widehat{\boldsymbol{\theta}}_n+\mathbf{h}/\sqrt n)
}{
L(\mathbf{x}_{1:n}\mid \widehat{\boldsymbol{\theta}}_n)
}.
\]

Before proving the local expansion, we record a simple localization fact. Since
$\boldsymbol{\theta}^*\in \operatorname{int}(\Theta)$ and
$\widehat{\boldsymbol{\theta}}_n\convp \boldsymbol{\theta}^*$, it follows that
$\widehat{\boldsymbol{\theta}}_n\in \operatorname{int}(\Theta)$ with probability tending to one.
Therefore, by the first-order condition for the score matching objective,
\[
\frac1n\sum_{i=1}^n
\mathbf g(\mathbf{x}_i,\widehat{\boldsymbol{\theta}}_n)
=
\mathbf 0
\]
with probability tending to one. Moreover, for any fixed $c>0$, setting
$r_n=c\log\sqrt n$,
\[
\sup_{\|\mathbf h\|\le r_n}\sup_{t\in[0,1]}
\left\|
\widehat{\boldsymbol{\theta}}_n+\frac{t\mathbf h}{\sqrt n}
-\boldsymbol{\theta}^*
\right\|
\le
\|\widehat{\boldsymbol{\theta}}_n-\boldsymbol{\theta}^*\|
+
\frac{r_n}{\sqrt n}
=o_p(1).
\]
Hence, with probability tending to one, all local points considered below lie in a neighborhood
of $\boldsymbol{\theta}^*$ on which the ETEL expansions are valid.

\begin{lemma}
\label{lem:bvm_local_expansion}
Let Assumptions~\ref{s_2}--\ref{s_4} hold. For any local value $\mathbf h$,
\begin{equation}
\label{eq:bvm_local_expansion}
\log R_n(\mathbf h)
=
-\frac12 \mathbf h^\top \boldsymbol V\mathbf h
+
O_p\left\{
\frac{\|\mathbf h\|+\|\mathbf h\|^2+\|\mathbf h\|^3}{\sqrt n}
\right\}.
\end{equation}
Consequently, for fixed $c>0$,
\begin{equation}
\label{eq:bvm_uniform_local_expansion}
\sup_{\|\mathbf{h}\|\leq c\log \sqrt{n}}
\left|
\log R_n(\mathbf{h})
+\frac{1}{2}\mathbf{h}^{\top}\boldsymbol{V}\mathbf{h}
\right|
=o_p(1).
\end{equation}
\end{lemma}

\begin{proof}
For a fixed $\mathbf h$, let $\boldsymbol{\theta}_t = \widehat{\boldsymbol{\theta}}_n+t\mathbf h/\sqrt n$, $t\in[0,1]$. By the localization fact above, $\boldsymbol{\theta}_t$ lies in a neighborhood of $\boldsymbol{\theta}^*$ with probability tending to one, uniformly over $t\in[0,1]$ and $\|\mathbf h\|\le c\log\sqrt n$. A Taylor expansion of $M_n(\boldsymbol{\theta})$ around
$\widehat{\boldsymbol{\theta}}_n$ gives
\begin{align}
\label{eq:bvm_taylor}
\log R_n(\mathbf h)
&=
\frac{1}{\sqrt n}
S_n(\widehat{\boldsymbol{\theta}}_n)^\top\mathbf h
+
\frac{1}{2n}
\mathbf h^\top
H_n(\boldsymbol{\theta}_t)
\mathbf h .
\end{align}
We next control the two terms on the right-hand side of
\eqref{eq:bvm_taylor}. Since $\widehat{\boldsymbol{\theta}}_n$ satisfies
\[
\frac1n\sum_{i=1}^n
\mathbf g(\mathbf{x}_i,\widehat{\boldsymbol{\theta}}_n)
=
\mathbf 0
\]
with probability tending to one, the standard ETEL multiplier expansion \citep{newey2004higher, chib2018bayesian} gives $\widehat{\boldsymbol{\lambda}}(\widehat{\boldsymbol{\theta}}_n)
= O_p(n^{-1/2})$. Under the exact empirical moment equation and local uniqueness of the ETEL multiplier, one may in fact take
$\widehat{\boldsymbol{\lambda}}(\widehat{\boldsymbol{\theta}}_n)=\mathbf 0$ with probability tending to one. In either case, the local ETEL derivative expansion implies
\begin{align}
\label{eq:bvm_linear_term}
\frac{1}{\sqrt n}
S_n(\widehat{\boldsymbol{\theta}}_n)^\top\mathbf h
&=
O_p(n^{-1/2}\|\mathbf h\|).
\end{align}

For the quadratic term, applying the same ETEL expansion along the local path $\boldsymbol{\theta}_t$ yields, uniformly over $\|\mathbf h\|\le c\log\sqrt n$,
\[
\widehat{\boldsymbol{\lambda}}(\boldsymbol{\theta}_t)
=
O_p\left(\frac{1+\|\mathbf h\|}{\sqrt n}\right),
\]
\[
\frac{d\widehat{\boldsymbol{\lambda}}(\boldsymbol{\theta}_t)}
{d\boldsymbol{\theta}^{\top}}
=
-\Delta^{-1}\Gamma
+
O_p\left(\frac{1+\|\mathbf h\|}{\sqrt n}\right),
\]
and
\[
\frac1n\sum_{i=1}^n
\mathbf g(\mathbf{x}_i,\boldsymbol{\theta}_t)
\mathbf g(\mathbf{x}_i,\boldsymbol{\theta}_t)^\top
=
\Delta
+
O_p\left(\frac{1+\|\mathbf h\|}{\sqrt n}\right).
\]
Therefore, the leading contribution to the normalized Hessian is
\[
-
\frac{d\widehat{\boldsymbol{\lambda}}(\boldsymbol{\theta}_t)^\top}
{d\boldsymbol{\theta}}
\Delta
\frac{d\widehat{\boldsymbol{\lambda}}(\boldsymbol{\theta}_t)}
{d\boldsymbol{\theta}^{\top}}
=
-\Gamma^\top\Delta^{-1}\Gamma
+
O_p\left(\frac{1+\|\mathbf h\|}{\sqrt n}\right).
\]
The remaining derivative terms in the Hessian contain either
$\widehat{\boldsymbol{\lambda}}(\boldsymbol{\theta}_t)$, centered empirical averages, or higher-order products of the local multiplier and are therefore of order $O_p((1+\|\mathbf h\|)/\sqrt n)$ after normalization. Hence,
\begin{align}
\label{eq:bvm_hessian_term}
\frac1n
\mathbf h^\top H_n(\boldsymbol{\theta}_t)\mathbf h
&=
-\mathbf h^\top\boldsymbol V\mathbf h
+
O_p\left(
\frac{(1+\|\mathbf h\|)\|\mathbf h\|^2}{\sqrt n}
\right),
\end{align}
where $\boldsymbol V=\Gamma^\top\Delta^{-1}\Gamma$. Substituting \eqref{eq:bvm_linear_term} and \eqref{eq:bvm_hessian_term} into \eqref{eq:bvm_taylor} gives
\[
\log R_n(\mathbf h)
=
-\frac12\mathbf h^\top\boldsymbol V\mathbf h
+
O_p\left(
\frac{\|\mathbf h\|}{\sqrt n}
+
\frac{(1+\|\mathbf h\|)\|\mathbf h\|^2}{\sqrt n}
\right),
\]
which proves \eqref{eq:bvm_local_expansion}.

Finally, the preceding expansion holds uniformly for
$\|\mathbf h\|\le c\log\sqrt n$. Therefore,
\[
\sup_{\|\mathbf h\|\le c\log\sqrt n}
\left|
\log R_n(\mathbf h)
+
\frac12\mathbf h^\top\boldsymbol V\mathbf h
\right|
=
O_p\left[
\frac{\log\sqrt n}{\sqrt n}
+
\frac{(\log\sqrt n)^2}{\sqrt n}
+
\frac{(\log\sqrt n)^3}{\sqrt n}
\right]
=o_p(1).
\]
This proves \eqref{eq:bvm_uniform_local_expansion}.
\end{proof}

\begin{lemma}
\label{lem:bvm_domain_decomposition}
Under Assumptions~\ref{s_1}--\ref{bvm_condition},
\begin{align}
\label{eq:bvm_l1_kernel}
\int
\Bigg|
\pi\left(\widehat{\boldsymbol{\theta}}_n+\frac{\mathbf{h}}{\sqrt n}\right)
R_n(\mathbf{h})
-
\pi(\boldsymbol{\theta}^*)
\exp\left(-\frac{1}{2}\mathbf{h}^{\top}
\boldsymbol{V}\mathbf{h}\right)
\Bigg|\,d\mathbf{h}
\convp 0 .
\end{align}
Here we use the convention that
\[
\pi\left(\widehat{\boldsymbol{\theta}}_n+\frac{\mathbf h}{\sqrt n}\right)
R_n(\mathbf h)=0
\]
whenever
$\widehat{\boldsymbol{\theta}}_n+\mathbf h/\sqrt n\notin\Theta$.
\end{lemma}

\begin{proof}
Fix $\delta>0$ small enough that $\left\{
\boldsymbol{\theta}:\|\boldsymbol{\theta}-\boldsymbol{\theta}^*\|<2\delta
\right\}$ is contained in the neighborhood in Assumption~\ref{s_3}. Since
$\widehat{\boldsymbol{\theta}}_n\convp\boldsymbol{\theta}^*$, it is enough to work on the event $E_n=
\left\{
\|\widehat{\boldsymbol{\theta}}_n-\boldsymbol{\theta}^*\|<\delta/2
\right\}$, whose probability tends to one. Let $r_n=c\log\sqrt n$, where $c>0$ is fixed. We split the domain of integration into
\[
A_{1n}=\{\|\mathbf h\|\le r_n\},\qquad
A_{2n}=\{r_n<\|\mathbf h\|\le \delta\sqrt n\},
\qquad
A_{3n}=\{\|\mathbf h\|>\delta\sqrt n\}.
\]

First consider $A_{1n}$. By Lemma~\ref{lem:bvm_local_expansion},
\[
\sup_{\mathbf h\in A_{1n}}
\left|
\log R_n(\mathbf h)
+
\frac12\mathbf h^\top\boldsymbol V\mathbf h
\right|
=o_p(1).
\]
Moreover,
\[
\sup_{\mathbf h\in A_{1n}}
\left\|
\widehat{\boldsymbol{\theta}}_n+\frac{\mathbf h}{\sqrt n}
-
\boldsymbol{\theta}^*
\right\|
\le
\|\widehat{\boldsymbol{\theta}}_n-\boldsymbol{\theta}^*\|
+
\frac{r_n}{\sqrt n}
=o_p(1).
\]
Therefore, by continuity of the prior at $\boldsymbol{\theta}^*$,
\[
\sup_{\mathbf h\in A_{1n}}
\left|
\pi\left(\widehat{\boldsymbol{\theta}}_n+\frac{\mathbf h}{\sqrt n}\right)
-
\pi(\boldsymbol{\theta}^*)
\right|
=o_p(1).
\]
Since $\boldsymbol V$ is positive definite by Assumption~\ref{s_4}, the preceding uniform expansion implies that for a fixed constant $C$, with probability tending to one,
\[
R_n(\mathbf h)
\le
C
\exp\left\{
-\frac12\mathbf h^\top\boldsymbol V\mathbf h
\right\},
\qquad \mathbf h\in A_{1n}.
\]
Hence the integrand over $A_{1n}$ is controlled by an integrable Gaussian envelope. Combining the uniform convergence of the prior and the uniform local likelihood-ratio expansion gives
\begin{align}
\int_{A_{1n}}
\Bigg|
\pi\left(\widehat{\boldsymbol{\theta}}_n+\frac{\mathbf h}{\sqrt n}\right)
R_n(\mathbf h)
-
\pi(\boldsymbol{\theta}^*)
\exp\left\{-\frac12\mathbf h^\top\boldsymbol V\mathbf h\right\}
\Bigg|\,d\mathbf h
\convp 0 .
\end{align}

Next consider $A_{2n}$. On $E_n$, for any $\mathbf h\in A_{2n}$ and $t\in[0,1]$,
\[
\left\|
\widehat{\boldsymbol{\theta}}_n+\frac{t\mathbf h}{\sqrt n}
-
\boldsymbol{\theta}^*
\right\|
\le
\|\widehat{\boldsymbol{\theta}}_n-\boldsymbol{\theta}^*\|
+
\frac{\|\mathbf h\|}{\sqrt n}
<
\frac{3\delta}{2}.
\]
Thus all points along the Taylor path lie inside the neighborhood on which the local ETEL Hessian expansion is valid. The uniform Hessian expansion gives, with probability tending to one,
\[
\frac1n H_n(\boldsymbol{\theta})
\preceq
-\frac12\boldsymbol V
\]
uniformly for $\|\boldsymbol{\theta}-\boldsymbol{\theta}^*\|<2\delta$. Combining this with the Taylor expansion around
$\widehat{\boldsymbol{\theta}}_n$, and using that the linear term is $O_p(n^{-1/2}\|\mathbf h\|)$, we obtain
\[
\log R_n(\mathbf h)
\le
-\frac14\mathbf h^\top\boldsymbol V\mathbf h,
\qquad \mathbf h\in A_{2n},
\]
with probability tending to one. Here the convention in the lemma makes the bound trivial for points satisfying
$\widehat{\boldsymbol{\theta}}_n+\mathbf h/\sqrt n\notin\Theta$.

Since the prior is bounded in a sufficiently small neighborhood of $\boldsymbol{\theta}^*$, there exists $C_\pi<\infty$ such that, with probability tending to one,
\[
\pi\left(\widehat{\boldsymbol{\theta}}_n+\frac{\mathbf h}{\sqrt n}\right)
R_n(\mathbf h)
\le
C_\pi
\exp\left\{
-\frac14\mathbf h^\top\boldsymbol V\mathbf h
\right\},
\qquad \mathbf h\in A_{2n}.
\]
Therefore, for some constant $C>0$,
\[
\int_{A_{2n}}
\pi\left(\widehat{\boldsymbol{\theta}}_n+\frac{\mathbf h}{\sqrt n}\right)
R_n(\mathbf h)\,d\mathbf h
\le
C_\pi
\int_{\|\mathbf h\|>r_n}
\exp(-C\|\mathbf h\|^2)\,d\mathbf h
=o_p(1).
\]
Similarly,
\[
\int_{A_{2n}}
\pi(\boldsymbol{\theta}^*)
\exp\left\{-\frac12\mathbf h^\top\boldsymbol V\mathbf h\right\}
\,d\mathbf h
=o(1).
\]
Hence the contribution from $A_{2n}$ to \eqref{eq:bvm_l1_kernel} is $o_p(1)$.

Finally consider $A_{3n}$. On $E_n$, if $\mathbf h\in A_{3n}$ and $\boldsymbol{\theta}_{\mathbf h} =
\widehat{\boldsymbol{\theta}}_n+\mathbf{h} / \sqrt{n}
\in\Theta$, then
\[
\left\|
\boldsymbol{\theta}_{\mathbf h}
-
\boldsymbol{\theta}^*
\right\|
\ge
\frac{\|\mathbf h\|}{\sqrt n}
-
\|\widehat{\boldsymbol{\theta}}_n-\boldsymbol{\theta}^*\|
>
\frac{\delta}{2}.
\]
By Assumption~\ref{bvm_condition}, there exists $\epsilon>0$ such that, with probability tending to one,
\[
\sup_{\|\boldsymbol{\theta}-\boldsymbol{\theta}^*\|>\delta/2}
\frac1n
\{M_n(\boldsymbol{\theta})-M_n(\boldsymbol{\theta}^*)\}
\le
-\epsilon .
\]
Moreover, since
$\widehat{\boldsymbol{\theta}}_n\convp\boldsymbol{\theta}^*$ and
$\widehat{\boldsymbol{\theta}}_n$ satisfies the empirical moment equation with probability tending to one, the local ETEL expansion at $\boldsymbol{\theta}^*$ gives
\[
M_n(\widehat{\boldsymbol{\theta}}_n)-M_n(\boldsymbol{\theta}^*)=O_p(1).
\]
Therefore, with probability tending to one, uniformly over
$\mathbf h\in A_{3n}$ such that $\boldsymbol{\theta}_{\mathbf h}\in\Theta$,
\[
\log R_n(\mathbf h)
=
\{M_n(\boldsymbol{\theta}_{\mathbf h})-M_n(\boldsymbol{\theta}^*)\}
-
\{M_n(\widehat{\boldsymbol{\theta}}_n)-M_n(\boldsymbol{\theta}^*)\}
\le
-\frac{n\epsilon}{2}.
\]
Thus, using the change of variables
$\boldsymbol{\theta}
=
\widehat{\boldsymbol{\theta}}_n+\mathbf h/\sqrt n$,
\begin{align*}
\int_{A_{3n}}
\pi\left(\widehat{\boldsymbol{\theta}}_n+\frac{\mathbf h}{\sqrt n}\right)
R_n(\mathbf h)\,d\mathbf h
&\le
e^{-n\epsilon/2}
\int_{\mathbb R^p}
\pi\left(\widehat{\boldsymbol{\theta}}_n+\frac{\mathbf h}{\sqrt n}\right)
\mathbbm{1}\left\{
\widehat{\boldsymbol{\theta}}_n+\frac{\mathbf h}{\sqrt n}\in\Theta
\right\}
\,d\mathbf h \\
&=
e^{-n\epsilon/2}n^{p/2}
\int_{\Theta}\pi(\boldsymbol{\theta})\,d\boldsymbol{\theta}
=o_p(1).
\end{align*}
The Gaussian contribution over $A_{3n}$ also satisfies
\[
\int_{A_{3n}}
\pi(\boldsymbol{\theta}^*)
\exp\left\{-\frac12\mathbf h^\top\boldsymbol V\mathbf h\right\}
\,d\mathbf h
\le
\pi(\boldsymbol{\theta}^*)
\int_{\|\mathbf h\|>\delta\sqrt n}
\exp\left\{-\frac12\mathbf h^\top\boldsymbol V\mathbf h\right\}
\,d\mathbf h
=o(1).
\]

Combining the bounds over $A_{1n}$, $A_{2n}$, and $A_{3n}$ proves
\eqref{eq:bvm_l1_kernel}.
\end{proof}

Then Theorem~\ref{bvm} follows directly.

\begin{proof}
Write $k_n(\mathbf h)
=
\pi\!\left(\widehat{\boldsymbol{\theta}}_n+{\mathbf h}/{\sqrt n}\right)
R_n(\mathbf h)$ with the convention that \(k_n(\mathbf h)=0\) whenever
\(\widehat{\boldsymbol{\theta}}_n+\mathbf h/\sqrt n\notin\Theta\). Let $C_n=\int_{\mathbb R^p} k_n(\mathbf h)\,d\mathbf h$. Then the posterior density of
\(\mathbf h=\sqrt n(\boldsymbol{\theta}-\widehat{\boldsymbol{\theta}}_n)\) is $q_n(\mathbf h\mid \mathbf x_{1:n})
=
C_n^{-1}k_n(\mathbf h)$.

By Lemma~\ref{lem:bvm_domain_decomposition},
\[
\int_{\mathbb R^p}
\left|
k_n(\mathbf h)
-
\pi(\boldsymbol{\theta}^*)
\exp\!\left\{-\frac12\mathbf h^\top\boldsymbol V\mathbf h\right\}
\right|d\mathbf h
\convp 0.
\]
Hence
\[
C_n
\convp
\pi(\boldsymbol{\theta}^*)
\int_{\mathbb R^p}
\exp\!\left\{-\frac12\mathbf h^\top\boldsymbol V\mathbf h\right\}
d\mathbf h
=
\pi(\boldsymbol{\theta}^*)(2\pi)^{p/2}|\boldsymbol V|^{-1/2}.
\]
The limiting constant is strictly positive by Assumptions~\ref{s_1}
and~\ref{s_4}. Therefore \(C_n^{-1}=O_p(1)\). Let
\[
\phi_{\boldsymbol V}(\mathbf h)
=
(2\pi)^{-p/2}|\boldsymbol V|^{1/2}
\exp\!\left\{-\frac12\mathbf h^\top\boldsymbol V\mathbf h\right\}.
\]
Combining \eqref{eq:bvm_l1_kernel} convergence  with \(C_n\convp
\pi(\boldsymbol{\theta}^*)(2\pi)^{p/2}|\boldsymbol V|^{-1/2}\), we obtain
\[
\int_{\mathbb R^p}
\left|
q_n(\mathbf h\mid \mathbf x_{1:n})
-
\phi_{\boldsymbol V}(\mathbf h)
\right|d\mathbf h
\convp 0.
\]
Therefore,
\[
\sup_{B\in\mathcal B(\mathbb R^p)}
\left|
\pi\!\left(
\sqrt n(\boldsymbol{\theta}-\widehat{\boldsymbol{\theta}}_n)\in B
\mid \mathbf x_{1:n}
\right)
-
\mathcal{N}(B;\mathbf 0,\boldsymbol V^{-1})
\right|
\le
\frac12
\int_{\mathbb R^p}
\left|
q_n(\mathbf h\mid \mathbf x_{1:n})
-
\phi_{\boldsymbol V}(\mathbf h)
\right|d\mathbf h
\convp 0.
\]
This proves \eqref{target}.
\end{proof}

\section{Supplementary Simulation Studies}
\label{app:simulation_res}
This section presents supplementary simulation studies evaluating SME-BETEL under both model misspecification and correct specification. We first illustrate the uncertainty calibration of SME-BETEL under misspecification through a simple one-dimensional example, followed by simulation studies of the autonormal model \citep{besag1974spatial} and truncated normal distribution. We then consider correctly specified models, including the autonormal model, Bingham distribution \citep{bingham1974antipodally}, and exponential graphical model \citep{yang2015graphical}, to evaluate the performance and efficiency of SME-BETEL when the working model is correctly specified. Together, these examples cover Euclidean, bounded, manifold, and nonnegative data supports and complement the mixed-domain preferential sampling study in the main text.

In all simulation studies, the Double Metropolis–Hastings (DMH) algorithm \citep{liang2010double} serves as the primary likelihood-based competitor. For the autonormal studies, we additionally include standard likelihood-based Bayesian inference, referred to as {\it Standard Bayes}, since the free-boundary joint likelihood is available in closed form \citep{balram1993noncausal}. Under model misspecification, coverage probabilities are computed with respect to each method’s corresponding target parameter, while under correct specification we report bias, mean squared error (MSE), mean absolute deviation (MAD), and coverage probability (CP). Unless otherwise specified, posterior computation is based on 10{,}000 MH iterations, with the first 3{,}000 discarded as burn-in. Details of the MH algorithm for SME-BETEL are provided in Appendix~\ref{app:sampling_algo}.

\subsection{Misspecified Models}
\label{subsec:misspecified_models}

\subsubsection{Illustration of Posterior Calibration Under Misspecification}
\label{subsec:calibration_example}

Suppose that the data-generating density is $p_x(x)=0.25f(0.25(x-1))$, where $f(u)=e^{-u}/(1+e^{-u})^2$ is the standard logistic density. Consider the working model $p(x;\theta)=f(x-\theta)$, equipped with the improper prior $\pi(\theta)\propto 1$. The working model is misspecified because its scale differs from that of the data-generating distribution. In this example, the target of usual Bayesian inference, defined as the minimizer of the KL divergence, is $\theta^\dagger=1$, while the score matching target defined in \eqref{theta^*} is also $\theta^*=1$. Thus, the two approaches target the same parameter, allowing us to directly compare their uncertainty quantification. Define $A=-\mathbb{E}_{p_x}[\partial^2\log p(X;\theta^\dagger)/\partial\theta^2]$ and $B=\mathbb{E}_{p_x}[\{\partial\log p(X;\theta^\dagger)/\partial\theta\}^2]$. Under model misspecification, the usual Bayesian posterior variance is asymptotically determined by $1/A$ \citep{ghosal2017fundamentals}, whereas the sampling variance of the maximum likelihood estimator is determined by the sandwich form $B/A^2$ \citep{white1982maximum}. Consequently, the asymptotic frequentist coverage of the nominal $95\%$ Bayesian credible interval is $P(|Z|<1.96\sqrt{A/B})\approx 0.56$, where $Z\sim N(0,1)$. Thus, the usual Bayesian $95\%$ credible interval has only about $56\%$ asymptotic frequentist coverage in this example. In contrast, equations~\eqref{target} and~\eqref{asymptotic_normal} imply that the limiting posterior variance of SME-BETEL matches the asymptotic sampling variance of the score matching estimator, so its $95\%$ credible interval is also an asymptotically valid $95\%$ frequentist confidence interval for $\theta^*$. To illustrate the finite-sample behavior, we conduct $500$ replications, each with sample size $n=50$. The empirical coverage probabilities of the nominal $95\%$ credible intervals are $0.55$ for the usual Bayesian approach and $0.92$ for SME-BETEL. The former closely agrees with the asymptotic coverage of approximately $0.56$, whereas SME-BETEL provides coverage much closer to the nominal level. This simple example illustrates the uncertainty calibration provided by SME-BETEL under model misspecification.

\subsubsection{Autonormal Model}
\label{subsec:misspecified_autonormal}

The autonormal model is a Gaussian Markov random field (GMRF) specified through Gaussian full conditional distributions, with conditional means given by linear combinations of neighboring sites \citep{besag1974spatial,rue2005gaussian}. Consider a second-order zero-mean GMRF $\mathbf{x} = (x_{ij})$ on an $M \times N$ lattice. For site $(i,j)$, let $n_h(i,j)$, $n_v(i,j)$, and $n_d(i,j)$ denote its horizontal, vertical, and diagonal neighbors, respectively; see Figure~\ref{fig:autonormal-neighbors}. The conditional density at site $(i,j)$, given all other sites, is
\begin{multline}
\label{autonormal_conditional_density}
        p\big(x_{ij} \mid \boldsymbol{\beta}, \sigma^2, x_{uv}; (u,v) \neq (i,j)\big) 
        = (2\pi \sigma^2)^{-1/2}
        \exp \bigg\{-\frac{1}{2\sigma^2}\bigg(x_{ij} 
        - \beta_h \sum_{(u,v) \in n_h(i,j)} x_{uv} \\
        - \beta_v \sum_{(u,v)\in n_v(i,j)} x_{uv} 
        - \beta_d \sum_{(u,v)\in n_d(i,j)} x_{uv}\bigg)^2 \bigg\},
\end{multline}
where $\boldsymbol{\beta} = (\beta_h, \beta_v, \beta_d)^{\top}$ and $\sigma^2$ are model parameters. Here, $\beta_h$, $\beta_v$, and $\beta_d$ control the strength of horizontal, vertical, and diagonal dependence, respectively. These conditional distributions determine a valid mean-zero joint Gaussian distribution, provided that the corresponding precision matrix is positive definite. The joint probability density of the process $\mathbf{x}$ \citep{besag1974spatial} can be written as
\begin{equation}
\label{autonormal_joint_density}
    p(\mathbf{x} ; \boldsymbol{\beta}, \sigma^2) 
= (2\pi \sigma^2)^{-MN/2} |\boldsymbol{B}(\boldsymbol{\beta})|^{1/2}
\exp \left\{-(2\sigma^2)^{-1}\mathbf{x}^{\top}\boldsymbol{B}(\boldsymbol{\beta})\mathbf{x}\right\},
\end{equation}
where $\boldsymbol{B}(\boldsymbol{\beta})$ is an $MN \times MN$ precision matrix depending on the parameters $\boldsymbol{\beta}$; its full formula is given in Appendix~\ref{subsec:precision_matrix_formula}. The determinant $|\boldsymbol{B}(\boldsymbol{\beta})|$ is generally difficult to evaluate, except in certain special cases \citep{besag1975estimation}. 


We now consider a misspecified working model. The data-generating model, denoted by $p_{\mathbf{x}}(\cdot)$, is the full autonormal model in \eqref{autonormal_joint_density} with parameters $\boldsymbol{\beta}=(\beta_h,\beta_v,\beta_d)^\top$, and throughout this simulation we fix $\sigma^2=1$. The working model, denoted by $p(\cdot;\boldsymbol{\alpha})$, is a constrained autonormal model with parameter $\boldsymbol{\alpha}=(\alpha_h,\alpha_d)^\top$, obtained by imposing $\beta_h=\beta_v=\alpha_h$ and $\beta_d=\alpha_d$. Equivalently, the working precision matrix $\boldsymbol{A}(\boldsymbol{\alpha})$ has the same form as $\boldsymbol{B}(\boldsymbol{\beta})$, but with the horizontal and vertical dependence parameters constrained to be equal. Thus, the working model is misspecified whenever $\beta_h\neq\beta_v$. Under misspecification, the likelihood-based Bayesian procedures and SME-BETEL target different pseudo-true parameters. The target of Standard Bayes and DMH, denoted by $\boldsymbol{\alpha}^{\dagger}$, is the minimizer of the KL divergence from $p_{\mathbf{x}}(\cdot)$ to $p(\cdot;\boldsymbol{\alpha})$, whereas the target of SME-BETEL, denoted by $\boldsymbol{\alpha}^{*}$, is the minimizer of the relative Fisher divergence in \eqref{theta^*}. In this example, both divergences admit closed-form expressions, which we use to numerically evaluate $\boldsymbol{\alpha}^{\dagger}$ and $\boldsymbol{\alpha}^{*}$. Coverage probabilities are then reported with respect to each method's corresponding target.

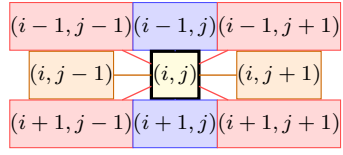
\begin{figure}[!t]
\centering
\begin{tikzpicture}[
  diagC/.style={draw=red!70, fill=red!15},
  vertC/.style={draw=blue!70, fill=blue!15},
  horzC/.style={draw=orange!80!black, fill=orange!15},
  ctrC/.style ={draw=black, very thick, fill=yellow!15},
  cell/.style ={minimum size=6.3mm, inner sep=0pt, outer sep=0pt, font=\scriptsize}
]
\matrix (m) [matrix of nodes, nodes={cell, draw=black!50},
             row sep=-\pgflinewidth, column sep=-\pgflinewidth] {
  |[diagC]| $(i-1,j-1)$ & |[vertC]| $(i-1,j)$ & |[diagC]| $(i-1,j+1)$ \\
  |[horzC]| $(i,j-1)$   & |[ctrC ]| $(i,j)$   & |[horzC]| $(i,j+1)$ \\
  |[diagC]| $(i+1,j-1)$ & |[vertC]| $(i+1,j)$ & |[diagC]| $(i+1,j+1)$ \\
};

\draw[orange!80!black, line width=0.5pt] (m-2-2) -- (m-2-1);
\draw[orange!80!black, line width=0.5pt] (m-2-2) -- (m-2-3);
\draw[blue!70, line width=0.5pt] (m-2-2) -- (m-1-2);
\draw[blue!70, line width=0.5pt] (m-2-2) -- (m-3-2);
\draw[red!70, line width=0.5pt] (m-2-2) -- (m-1-1);
\draw[red!70, line width=0.5pt] (m-2-2) -- (m-1-3);
\draw[red!70, line width=0.5pt] (m-2-2) -- (m-3-1);
\draw[red!70, line width=0.5pt] (m-2-2) -- (m-3-3);
\end{tikzpicture}
\caption{Second-order neighborhood of site $(i,j)$ in the autonormal model. Colors indicate horizontal, vertical, and diagonal neighbors.}
\label{fig:autonormal-neighbors}
\end{figure}
 
For implementation, we use the free-boundary version of the autonormal model, under which boundary sites have fewer neighbors than interior sites. In this setting, \citet{balram1993noncausal} derived an equivalent likelihood representation for the joint density in \eqref{autonormal_joint_density} that is convenient for evaluating likelihood contributions; see Appendix~\ref{subsec:autonormal_likelihood}. We use this representation for the likelihood-based methods and use the same free-boundary structure when constructing the SME-BETEL moment functions. For the constrained working model, we use the prior $\pi(\boldsymbol{\alpha})\propto \mathbbm{1}(2|\alpha_h|+2|\alpha_d|<0.5)$, corresponding to the stationarity condition under the constrained model \citep{balram1993noncausal}. To implement SME-BETEL, we construct moment functions from block-level observations. We cannot directly use site-level moment functions $\mathbf{g}(x_{ij},\boldsymbol{\alpha})$ defined in \eqref{population_moment_condition}, because the lattice sites are spatially dependent, whereas the theoretical justification of SME-BETEL is based on independent observations; see Assumption~\ref{s_3}(a). We therefore adopt a blockwise working assumption: after partitioning the lattice into non-overlapping blocks, the block-level observations are treated as independent. Assuming $M=N$, we partition the lattice evenly into $n$ non-overlapping blocks. Specifically, for block observations $\mathbf{x}_1,\ldots,\mathbf{x}_n$, we define $\mathbf{g}(\mathbf{x}_i,\boldsymbol{\alpha})$, $i=1,\ldots,n$, using the block-level working precision matrix $\widetilde{\boldsymbol{A}}(\boldsymbol{\alpha})$, whose dimension matches the block size and whose entries follow the same local neighborhood structure as $\boldsymbol{A}(\boldsymbol{\alpha})$. The resulting SME-BETEL posterior is $\pi(\boldsymbol{\alpha} \mid \mathbf{x}) \propto L(\boldsymbol{\alpha})\pi(\boldsymbol{\alpha})$, where $L(\boldsymbol{\alpha})$ is the ETEL function defined in \eqref{likelihood}.

For the simulation setup, we take $M=N\in\{20,30,40,50,60,70\}$. For each lattice size, we generate 500 independent replications with $(\beta_h^*,\beta_v^*,\beta_d^*)=(0.1,0.2,0.05)$. Each replication consists of one spatial field on an $M\times N$ lattice, so increasing $M$ increases the dimension of the observed field rather than the number of independent fields. For each method and each replication, we construct 95\% credible intervals and compute coverage with respect to the corresponding pseudo-true target. The results are summarized in Figure~\ref{fig:simres_autonormal_robustness}. As the field size increases, the coverage probabilities of Standard Bayes and DMH decrease substantially for both parameters, especially for $\alpha_h$. In contrast, SME-BETEL maintains coverage close to the nominal level across all field sizes. This pattern is consistent with the theoretical discussion in Section~\ref{sec:sme_betel}, where SME-BETEL is shown to provide calibrated uncertainty quantification under misspecification. The lower coverage for $\alpha_h$ relative to $\alpha_d$ is also expected from the construction of the working model, since the restriction $\beta_h=\beta_v$ forces $\alpha_h$ to summarize both horizontal and vertical dependence. Table~\ref{tab:simres_mis_autonormal_robustness} reports the specific coverage probabilities.

\begin{figure}[!t]
\centering
\includegraphics[width=0.485\textwidth]{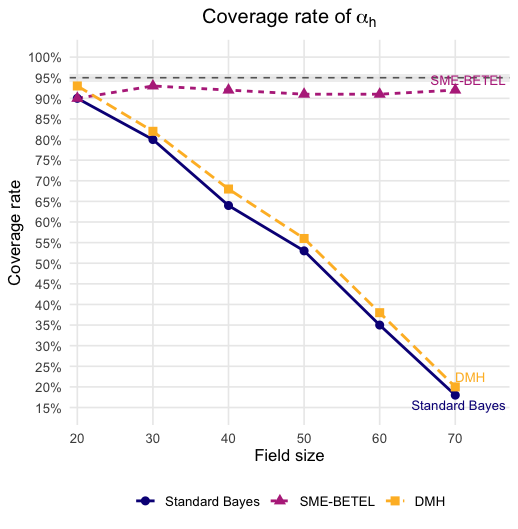}
\hfill
\includegraphics[width=0.485\textwidth]{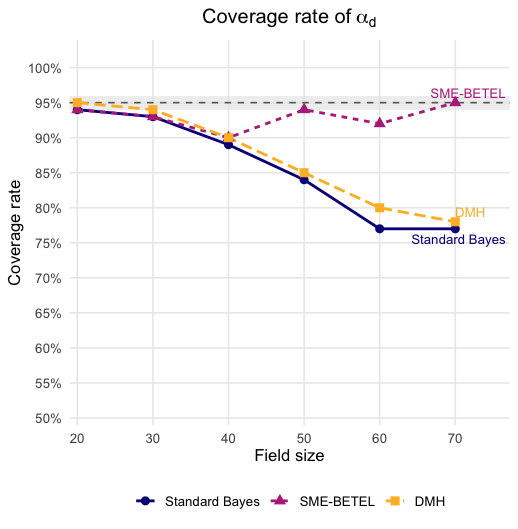}
\caption{Coverage probabilities of 95\% credible intervals under model misspecification for $\alpha_h$ and $\alpha_d$ across field sizes.}
\label{fig:simres_autonormal_robustness}
\end{figure}

\begin{table}[!htbp]
\centering
\small
\setlength{\tabcolsep}{6pt}
\renewcommand{\arraystretch}{1.12}
\caption{Simulation results for the misspecified autonormal model. Entries are empirical coverage rates; number in parentheses gives mean interval widths.}
\label{tab:simres_mis_autonormal_robustness}
\begin{tabular}{@{}llccc@{}}
\toprule
\textbf{Field size ($M=N$)} & \textbf{Parameter} &
\textbf{Standard Bayes} & \textbf{SME--BETEL} & \textbf{DMH} \\
\midrule
\multirow{2}{*}{$M=20$}
  & $\alpha_h$ & 90.0 (0.11) & 90.0 (0.11) & \textbf{93.0} (0.11) \\
  & $\alpha_d$ & 94.0 (0.12) & 94.0 (0.13) & \textbf{95.0} (0.13) \\
\midrule
\multirow{2}{*}{$M=30$}
  & $\alpha_h$ & 80.0 (0.08) & \textbf{92.0} (0.08) & 82.0 (0.08) \\
  & $\alpha_d$ & 93.0 (0.09) & 93.0 (0.10) & \textbf{94.0} (0.10) \\
\midrule
\multirow{2}{*}{$M=40$}
  & $\alpha_h$ & 64.0 (0.06) & \textbf{91.0} (0.06) & 68.0 (0.06) \\
  & $\alpha_d$ & 89.0 (0.07) & \textbf{90.0} (0.07) & \textbf{90.0} (0.07) \\
\midrule
\multirow{2}{*}{$M=50$}
  & $\alpha_h$ & 53.0 (0.05) & \textbf{90.0} (0.05) & 56.0 (0.05) \\
  & $\alpha_d$ & 84.0 (0.05) & \textbf{93.0} (0.06) & 85.0 (0.06) \\
\midrule
\multirow{2}{*}{$M=60$}
  & $\alpha_h$ & 35.0 (0.04) & \textbf{91.0} (0.04) & 38.0 (0.04) \\
  & $\alpha_d$ & 77.0 (0.04) & \textbf{92.0} (0.05) & 80.0 (0.05) \\
\midrule
\multirow{2}{*}{$M=70$}
  & $\alpha_h$ & 18.0 (0.03) & \textbf{92.0} (0.03) & 20.0 (0.04) \\
  & $\alpha_d$ & 77.0 (0.04) & \textbf{95.0} (0.04) & 78.0 (0.04) \\
\bottomrule
\end{tabular}
\end{table}

\subsubsection{Truncated Normal Distribution}
\label{truncated_normal}

This subsection reports a supplementary misspecified example on a bounded support. The purpose is to evaluate whether SME-BETEL provides calibrated uncertainty when the working model ignores dependence in a truncated bivariate normal distribution.

Let $\mathcal{D}\subset\mathbb{R}^2$ denote the truncation region. The
data-generating density is
\[
p_{\mathbf{x}}(\mathbf{x})
=
Z(\boldsymbol{\mu}_0,\Sigma)^{-1}
(2\pi)^{-1}|\Sigma|^{-1/2}
\exp\left\{
-\frac{1}{2}
(\mathbf{x}-\boldsymbol{\mu}_0)^{\top}
\Sigma^{-1}
(\mathbf{x}-\boldsymbol{\mu}_0)
\right\},
\qquad \mathbf{x}\in\mathcal{D},
\]
where $Z(\boldsymbol{\mu}_0,\Sigma)$ is the normalizing constant over
$\mathcal{D}$. The working model is
\[
p(\mathbf{x};\boldsymbol{\theta})
=
Z(\boldsymbol{\theta})^{-1}
(2\pi)^{-1}|\Lambda|^{-1/2}
\exp\left\{
-\frac{1}{2}
(\mathbf{x}-\boldsymbol{\theta})^{\top}
\Lambda^{-1}
(\mathbf{x}-\boldsymbol{\theta})
\right\},
\qquad \mathbf{x}\in\mathcal{D},
\]
where $\Lambda=\mathbf{I}_2$ is fixed and $\boldsymbol{\theta}$ is the parameter of interest. The working model is misspecified whenever $\Sigma\ne \mathbf{I}_2$; the examples below use covariance matrices with nonzero off-diagonal entries. Under misspecification, the likelihood-based method and SME-BETEL target different pseudo-true parameters. The pseudo-true parameter $\boldsymbol{\theta}^{\dagger}$ for DMH is defined as the minimizer of $\mathrm{D}_{\rm KL}\!\left(p_{\mathbf{x}}(\cdot) \,\|\, p(\cdot; \boldsymbol{\theta})\right)$, which admits a closed-form expression in this example:
\begin{align*}
\mathrm{D}_{\rm KL}\!\left(p_{\mathbf{x}}(\cdot) \,\|\, p(\cdot; \boldsymbol{\theta})\right)
&= \log \frac{Z(\boldsymbol{\theta})}{Z(\boldsymbol{\mu}_0, \Sigma)} 
   - \frac{1}{2} \Bigl\{
      \mathrm{tr}(\Sigma^{-1}\Omega) 
      + (\boldsymbol{\nu} - \boldsymbol{\mu}_0)^{\top}
        \Sigma^{-1}(\boldsymbol{\nu} - \boldsymbol{\mu}_0)
     \Bigr\} \\[4pt]
&\quad + \frac{1}{2} \Bigl\{
      \mathrm{tr}(\Lambda^{-1}\Omega)
      + (\boldsymbol{\nu} - \boldsymbol{\theta})^{\top}
        \Lambda^{-1}(\boldsymbol{\nu} - \boldsymbol{\theta})
     \Bigr\},
\end{align*}
where $\mathbb{E}_{p_{\mathbf{x}}}[\mathbf{x}] = \boldsymbol{\nu}$ and $\mbox{Var}_{p_{\mathbf{x}}}(\mathbf{x}) = \Omega$. Here, we use the \texttt{R} function \texttt{tmvtnorm::mtmvnorm()} to get $\boldsymbol{\nu}$ and $\Omega$, and use the trapezoidal rule to numerically approximate $Z(\boldsymbol{\mu}_0, \Sigma)$ and $Z(\boldsymbol{\theta})$. Then $\boldsymbol{\theta}^{\dagger}$ is straightforward to compute. The pseudo-true parameter of SME-BETEL is $\boldsymbol{\theta}^{*}$ such that $\boldsymbol{\theta}^{*} = \argmin_{\boldsymbol{\theta}} \mathrm{D}_{\rm F}\!\left(p_{\mathbf{x}}(\cdot) \,\|\, p(\cdot; \boldsymbol{\theta})\right)$. In this example,
\begin{align*}
    \mathrm{D}_{\rm F}\!\left(p_{\mathbf{x}}(\cdot) \,\|\, p(\cdot; \boldsymbol{\theta})\right) &= \frac{1}{2} \int h||\psi(\mathbf{x}; \boldsymbol{\theta}) - \psi_{\mathbf{x}}(\mathbf{x})||^2 p_{\mathbf{x}}(\mathbf{x}) d\mathbf{x} \\
    &= \mathbb{E}_{p_{\mathbf{x}}}\big[\frac{1}{2}h||\Lambda^{-1}(\boldsymbol{\theta} - \mathbf{x}) - \Sigma^{-1}(\boldsymbol{\mu}_0 - \mathbf{x})||^2\big] \hspace{5mm} \text{(omit constant terms)},
\end{align*}
where $h$ is the weight function. Unlike $\mathrm{D}_{\rm KL}\!\left(p_{\mathbf{x}}(\cdot) \,\|\, p(\cdot; \boldsymbol{\theta})\right)$, $\mathrm{D}_{\rm F}\!\left(p_{\mathbf{x}}(\cdot) \,\|\, p(\cdot; \boldsymbol{\theta})\right)$ does not admit an analytical form due to the existence of the weight function $h$. Fortunately, it can be easily approximated.

We consider two misspecified cases. In Case~I, the truncation region is symmetric: $\mathcal{D}=[-2,2]^2$, $\boldsymbol{\mu}_0=(0,0)^{\top}$ and $\Sigma=
\begin{pmatrix}
7 & 1.2 \\
1.2 & 7
\end{pmatrix}$. In Case~II, the truncation region induces a more asymmetric setting: $\mathcal{D}=[-1,3]^2$, $\boldsymbol{\mu}_0=(2.5,2.6)^{\top}$ and $\Sigma=
\begin{pmatrix}
5 & 1.2 \\
1.2 & 5
\end{pmatrix}$. For both methods, we use the prior $\boldsymbol{\theta}\sim \mathrm{N}(\mathbf{0},5\mathbf{I}_2)$. Since the density is defined on a two-dimensional bounded domain, SME-BETEL uses the bounded domain score matching objective in \citet{liu2022estimating}. The weight function is chosen as in \eqref{mollifier} with $\epsilon=0.2$. Table~\ref{tab:truncnorm_scenarios} reports the coverage probabilities of nominal 95\% credible intervals. Coverage is computed with respect to the corresponding pseudo-true target of each method. Across both cases and all sample sizes, SME-BETEL achieves coverage close to the nominal level, whereas DMH consistently undercovers. This supports the robustness of SME-BETEL under bounded domain model misspecification.

\begin{table}[!t]
\centering
\caption{Coverage probabilities of 95\% credible intervals for the misspecified truncated normal model.}
\label{tab:truncnorm_scenarios}
\begingroup
\small
\setlength{\tabcolsep}{4pt}
\renewcommand{\arraystretch}{0.92}
\begin{tabular}{@{}llcccc@{}}
\toprule
& & \multicolumn{2}{c}{Case~I} & \multicolumn{2}{c}{Case~II} \\
\cmidrule(lr){3-4}\cmidrule(l){5-6}
Sample size & Parameter & DMH & SME-BETEL & DMH & SME-BETEL \\
\midrule
$200$  & $\theta_1$ & 88.0 & \textbf{94.0} & 88.0 & \textbf{96.0} \\
       & $\theta_2$ & 88.0 & \textbf{94.0} & 88.0 & \textbf{94.0} \\
\addlinespace[1pt]
$1000$ & $\theta_1$ & 87.0 & \textbf{95.0} & 90.0 & \textbf{95.0} \\
       & $\theta_2$ & 90.0 & \textbf{96.0} & 88.0 & \textbf{95.0} \\
\addlinespace[1pt]
$3000$ & $\theta_1$ & 85.0 & \textbf{94.0} & 89.0 & \textbf{95.0} \\
       & $\theta_2$ & 90.0 & \textbf{95.0} & 86.0 & \textbf{95.0} \\
\bottomrule
\end{tabular}
\endgroup
\end{table}

\subsection{Correctly Specified Models}
\label{subsec:correctly_specified_models}

\subsubsection{Autonormal Model}
\label{subsec:correct_autonormal}

This subsection reports the correctly specified counterpart of the autonormal experiment in Section~\ref{subsec:misspecified_autonormal}. The purpose is to assess the efficiency of SME-BETEL relative to likelihood-based Bayesian methods when the working model is correctly specified.

The data-generating model and the working model are both the full second-order autonormal model in \eqref{autonormal_joint_density}, with parameter $\boldsymbol{\beta} = (\beta_h,\beta_v,\beta_d)^\top$. We set $M=N=40$, fix $\sigma^2=1$, and use the true parameter value $\boldsymbol{\beta}^* = (\beta_h^*,\beta_v^*,\beta_d^*)^\top =
(0.08,0.15,0.02)^\top$. For each replication, the lattice is initialized by drawing each site independently from $\mathrm{N}(0,1)$, and one spatial field is generated by running 5,000 full Gibbs sweeps using the conditional density in \eqref{autonormal_conditional_density}. We generate 500 independent replications. We compare Standard Bayes, DMH, and SME-BETEL using the same prior $\pi(\boldsymbol{\beta})
\propto
\mathbbm{1}
\left\{
|\beta_h|+|\beta_v|+2|\beta_d|<0.5
\right\}$, which corresponds to the stationarity condition for the full autonormal model. The posterior samplers are initialized at $(0.2,0.1,0.01)^\top$. 

Since the model is correctly specified, all three methods target the true parameter
$\boldsymbol{\beta}^{*}$. Table~\ref{tab:simres_autonormal} reports the bias, MAD, MSE, and CP for each component of $\boldsymbol{\beta}$. All three methods achieve coverage probabilities close to the nominal 95\% level. Standard Bayes and DMH have slightly smaller MAD and MSE than SME-BETEL, reflecting the efficiency advantage
of likelihood-based inference under correct specification. SME-BETEL remains competitive, although its posterior is slightly more dispersed because the score matching estimating equations do not use the information contained in the normalizing constant.

\begin{table}[!t]
\centering
\caption{Simulation results for the correctly specified autonormal model.}
\label{tab:simres_autonormal}
\begingroup
\small
\setlength{\tabcolsep}{4pt}
\renewcommand{\arraystretch}{0.92}
\begin{tabular}{@{}ll
S[table-format=-1.1]
S[table-format=1.1]
S[table-format=1.2]
S[table-format=2.1]
@{}}
\toprule
Parameter & Model
& \multicolumn{1}{c}{Bias $(\times 10^2)$}
& \multicolumn{1}{c}{MAD $(\times 10^2)$}
& \multicolumn{1}{c}{MSE $(\times 10^2)$}
& \multicolumn{1}{c}{CP} \\
\midrule
$\beta_h$ & Standard Bayes &  0.2 & 1.8 & 0.05 & 96.0 \\
          & DMH            &  0.2 & 1.9 & 0.05 & 96.0 \\
          & SME-BETEL      &  0.5 & 2.1 & 0.07 & 94.0 \\
\addlinespace[1pt]
$\beta_v$ & Standard Bayes & -0.5 & 1.9 & 0.05 & 95.0 \\
          & DMH            & -0.2 & 1.9 & 0.05 & 96.0 \\
          & SME-BETEL      & -0.2 & 2.1 & 0.07 & 94.0 \\
\addlinespace[1pt]
$\beta_d$ & Standard Bayes & -0.2 & 1.3 & 0.03 & 95.0 \\
          & DMH            & -0.1 & 1.3 & 0.03 & 96.0 \\
          & SME-BETEL      & -0.1 & 1.7 & 0.04 & 95.0 \\
\bottomrule
\end{tabular}
\endgroup
\end{table}

\subsubsection{Bingham Distribution}
\label{bingham_dist}

This subsection reports a supplementary correctly specified example for spherical data. The purpose is to evaluate the performance of SME-BETEL when the support is a smooth manifold and the working model is correctly specified. The Bingham distribution \citep{bingham1974antipodally,mardia2009directional} is commonly used for antipodally symmetric data on the unit sphere $\mathbb{S}^{m-1}
=
\left\{
\mathbf{x}\in\mathbb{R}^m:
x_1^2+\cdots+x_m^2=1
\right\}$. Its density on $\mathbb{S}^{m-1}$ is
\[
p(\mathbf{x};\boldsymbol{A})
=
z(\boldsymbol{A})^{-1}
\exp\{-\mathbf{x}^{\top}\boldsymbol{A}\mathbf{x}\},
\qquad
\mathbf{x}\in\mathbb{S}^{m-1},
\]
where $\boldsymbol{A}$ is an $m\times m$ symmetric matrix and
$z(\boldsymbol{A})$ is the normalizing constant. By rotating to the principal axes, it is sufficient to consider a diagonal parameter matrix. The density can then be written as
\[
p(\mathbf{x};\boldsymbol{\lambda})
\propto
\exp\left\{
-\sum_{i=1}^{m}\lambda_i x_i^2
\right\}.
\]
We take $m=3$. For identifiability, we impose the conventional constraint $\lambda_1\geq \lambda_2\geq \lambda_3=0$ \citep{kent1987asymptotic}. Thus, the unnormalized density becomes
\[
p(\mathbf{x};\lambda_1,\lambda_2)
\propto
\exp\{-\lambda_1 x_1^2-\lambda_2 x_2^2\},
\qquad
\mathbf{x}\in\mathbb{S}^{2}.
\]
Since $\mathbb{S}^2$ is a smooth manifold without boundary, SME-BETEL uses the manifold score matching objective of \citet{mardia2016score} to construct the estimating equations. To enforce the ordering constraint, we use the reparameterization $\theta_1=\lambda_1-\lambda_2$ and $\theta_2=\lambda_2$ so that $\theta_1,\theta_2\geq 0$ and $\lambda_1=\theta_1+\theta_2$ and $\lambda_2=\theta_2$.

We assign independent $\mathrm{Exp}(1)$ priors to $\theta_1$ and $\theta_2$. In the MH sampler, the transformed parameters
$(\log\theta_1,\log\theta_2)$ are updated using a normal random-walk proposal. The true parameter is set to $(\lambda_1^*,\lambda_2^*)=(2,1.5)$ which corresponds to $(\theta_1^*,\theta_2^*)=(0.5,1.5)$. Samples are generated using the \textsc{simdd} package in \texttt{R}. We generate 500 replications, each with 1000 observations. The initial value is
set to $(\theta_1,\theta_2)=(1,1)$. Table~\ref{tab:simres_Bingham_3d} summarizes the simulation results. Since the model is correctly specified, both DMH and SME-BETEL achieve negligible bias and coverage probabilities close to the nominal 95\% level for both $\lambda_1$ and $\lambda_2$. DMH has slightly smaller MAD and MSE values,
reflecting the efficiency advantage of likelihood-based inference under correct specification. SME-BETEL remains competitive, with only modestly larger variability.

\begin{table}[!t]
\centering
\caption{Simulation results for the correctly specified Bingham distribution.}
\label{tab:simres_Bingham_3d}
\begingroup
\small
\setlength{\tabcolsep}{4pt}
\renewcommand{\arraystretch}{0.92}
\begin{tabular}{@{}ll
S[table-format=-1.1]
S[table-format=2.1]
S[table-format=1.1]
S[table-format=2.1]
@{}}
\toprule
Parameter & Model
& \multicolumn{1}{c}{Bias $(\times 10^2)$}
& \multicolumn{1}{c}{MAD $(\times 10^2)$}
& \multicolumn{1}{c}{MSE $(\times 10^2)$}
& \multicolumn{1}{c}{CP} \\
\midrule
$\lambda_1$ & DMH       & -1.8 & 10.9 & 1.8 & 94.0 \\
            & SME-BETEL & -1.6 & 11.0 & 1.9 & 94.0 \\
\addlinespace[1pt]
$\lambda_2$ & DMH       & -0.4 &  9.8 & 1.5 & 95.0 \\
            & SME-BETEL & -0.1 & 10.0 & 1.5 & 95.0 \\
\bottomrule
\end{tabular}
\endgroup
\end{table}

\subsubsection{Exponential Graphical Model}
\label{subsec:exp_graphical}

Consider a random vector $\mathbf{x}=(x_1,\ldots,x_m)$, where each component $x_i$ takes values on the positive real line. Let $G=(V,E)$ denote an undirected graph with node set $V := \{1, \ldots, m\}$ indexing the variables $\{x_i\}_{i=1}^m$. A subset $C \subset V$ is called a clique if all of its vertices are mutually connected. Denote by $\mathcal{M}$ the collection of all cliques in $G$, and let $\{\phi_C(x_C)\}_{C \in \mathcal{M}}$ be a family of compatibility functions, where $x_C=\{x_j:j\in C\}$. A probability distribution with respect to $G$ can then be written as
\begin{equation}
\label{joint_dist}
    p(x_1, \cdots, x_m) = \frac{1}{z}\prod_{C \in \mathcal{M}} \phi_C(x_C) 
\end{equation}
where $z$ is the normalizing constant. When each node-conditional distribution $p(x_r\mid x_{-r})$ belongs to a univariate exponential family, the compatible joint distribution takes the form in \eqref{joint_dist}. This class of models is known as exponential family graphical models \citep{yang2015graphical}. We now consider the correctly specified exponential graphical model, a special case of exponential family graphical models. For $m=3$, the unnormalized density is given by
\begin{equation}
\label{egm}
    p(x_1, x_2, x_3) \hspace{1mm} \propto \hspace{1mm} \exp \bigg\{-\big(\theta_1 x_1 + \theta_2 x_2 + \theta_3 x_3 + \phi (x_1 x_2 + x_2 x_3 + x_1 x_3)\big)\bigg\},
\end{equation}
where $x_1, x_2, x_3 \geq 0$, $\boldsymbol{\theta} = (\theta_1, \theta_2, \theta_3) > \mathbf{0}_3$, and $\phi > 0$. In this example, we fix $\phi = 0.01$. Since the data are nonnegative, we derive the moment function $\mathbf{g}(\mathbf{x}, \boldsymbol{\theta})$ using the score matching objective function defined for nonnegative data in \citet{hyvarinen2007some}. For posterior sampling, we place independent $\mbox{Exp}(0.2)$ priors on each component of $\boldsymbol{\theta}$. Since the parameters are positive, we update $\log \boldsymbol{\theta}$ using a normal random-walk proposal in the MH algorithm. We set $(\theta^*_1, \theta^*_2, \theta^*_3) = (5, 4, 4)$, and then generate 500 replications, each consisting of 1000 independent three-dimensional observations generated from the model in \eqref{egm}. The initial value of $\boldsymbol{\theta}$ is $(3, 3, 3)$, and the simulation results are summarized in Table~\ref{tab:simres_EGM}. Again, when the model is correctly specified, the performances of DMH and SME-BETEL are similar to those observed for the autonormal model and the Bingham distribution. Taken together, these simulation studies show that SME-BETEL remains competitive with likelihood-based Bayesian methods under correct specification, while providing substantially more robust inference under model misspecification.

\begin{table}[!t]
\centering
\caption{Simulation results for the exponential graphical model under correct specification.}
\label{tab:simres_EGM}
\begingroup
\small
\setlength{\tabcolsep}{4pt}
\renewcommand{\arraystretch}{0.92}
\begin{tabular}{@{}ll
                S[table-format=+1.1]
                S[table-format=2.1]
                S[table-format=1.1]
                S[table-format=2.1]
                @{}}
\toprule
Parameter & Model
& \multicolumn{1}{c}{Bias $(\times 10^2)$}
& \multicolumn{1}{c}{MAD $(\times 10^2)$}
& \multicolumn{1}{c}{MSE $(\times 10^2)$}
& \multicolumn{1}{c}{CP} \\
\midrule
\multirow{2}{*}{$\theta_{1}$}
  & DMH       & -0.1 & 12.8 & 2.6 & 95.0 \\
  & SME-BETEL &  0.1 & 17.8 & 5.1 & 92.0 \\
\addlinespace[1pt]
\multirow{2}{*}{$\theta_{2}$}
  & DMH       & -0.2 & 10.3 & 1.6 & 96.0 \\
  & SME-BETEL & -0.9 & 13.5 & 2.8 & 94.0 \\
\addlinespace[1pt]
\multirow{2}{*}{$\theta_{3}$}
  & DMH       &  0.2 & 10.7 & 1.7 & 94.0 \\
  & SME-BETEL &  0.3 & 13.4 & 2.8 & 95.0 \\
\bottomrule
\end{tabular}
\endgroup
\end{table}

\section{Additional Formulas and Implementation Details}
\label{app:formulas_computation_details}

This section collects model-specific formulas and implementation details used in the simulation studies and data application. Section \ref{subsec:weight_function} specifies the weight function used for bounded domain score matching. Section \ref{subsec:additional_formula_autonormal} provides the precision matrix and free-boundary likelihood representation for the autonormal model, while Section \ref{subsec:kernel_form} describes the kernel convolution approximation used for the $\mathrm{Mat\acute{e}rn}$ Gaussian processes in the preferential sampling model.

\subsection{Weight Functions for Bounded Domains}
\label{subsec:weight_function}
For the mixed-domain score matching objective introduced in Section~\ref{sec:support_adapted_sm}, the weight function $h$ is required to vanish on the boundary of the compact domain. In the numerical studies we use a smooth cubic weight function that equals one away from the boundary and tapers to zero within a distance $\epsilon$ of the boundary. For a one-dimensional bounded domain $\mathcal{D} = [a, b]$, we define
\begin{equation}
\label{mollifier}
h(x)=
\begin{cases}
3\left(\dfrac{x-a}{\epsilon}\right)^2
-2\left(\dfrac{x-a}{\epsilon}\right)^3,
& a \le x < a+\epsilon,\\
1, & a+\epsilon \le x \le b-\epsilon,\\
3\left(\dfrac{b-x}{\epsilon}\right)^2
-2\left(\dfrac{b-x}{\epsilon}\right)^3,
& b-\epsilon < x \le b .
\end{cases}
\end{equation}
Here, $\epsilon$ is fixed and controls the width of the boundary region over which the weight function smoothly tapers to zero. For a two-dimensional rectangular domain, the weight function takes the form $h_1(x_1)h_2(x_2)$, where $h_1$ and $h_2$ are defined analogously to $h$ for the corresponding coordinate domains.

\subsection{Autonormal Model Formulas}
\label{subsec:additional_formula_autonormal}
\subsubsection{Precision Matrix}
\label{subsec:precision_matrix_formula}

For sites \(s=(i,j)\) and \(t=(u,v)\), the entries of \(\boldsymbol{B}(\boldsymbol{\beta})\) are
\[
B_{s,t}(\boldsymbol{\beta}) =
\begin{cases}
1, & s=t,\\
-\beta_h, & s \text{ and } t \text{ are horizontal neighbors},\\
-\beta_v, & s \text{ and } t \text{ are vertical neighbors},\\
-\beta_d, & s \text{ and } t \text{ are diagonal neighbors},\\
0, & \text{otherwise}.
\end{cases}
\]
The conditional specification defines a proper joint Gaussian distribution when
\(\boldsymbol{B}(\boldsymbol{\beta})\) is positive definite.

\subsubsection{Free-Boundary Likelihood Representation}
\label{subsec:autonormal_likelihood}

Under the free-boundary condition, \citet{balram1993noncausal} showed that the joint probability density defined in \eqref{autonormal_joint_density} is equivalent to
\begin{equation*}
p(\mathbf{x} ; \boldsymbol{\beta}, \sigma^2) \hspace{1mm} \propto \hspace{1mm} (\sigma^2)^{-MN/2}|\boldsymbol{B}(\boldsymbol{\beta})|^{1/2}\exp \big\{-\frac{MN}{2\sigma^2}(S_x - 2\beta_h X_h - 2\beta_v X_v - 2\beta_d X_d)\big\} 
\end{equation*}
where
\begin{align*}
S_x &= \frac{1}{MN} \sum_{i=1}^M \sum_{j=1}^N x_{ij}^2,
\qquad
X_h = \frac{1}{MN} \sum_{i=1}^M \sum_{j=1}^{N-1} x_{ij}x_{i, j+1}, \\
X_v &= \frac{1}{MN} \sum_{i=1}^{M-1} \sum_{j=1}^N x_{ij} x_{i+1, j},
\qquad
X_d = \frac{1}{MN}\!\left(\sum_{i=1}^{M-1} \sum_{j=1}^{N-1} x_{ij}x_{i+1, j+1}
+ \sum_{i=1}^{M-1} \sum_{j=2}^N x_{ij}x_{i+1, j-1}\right).
\end{align*}
This representation is used to evaluate the likelihood-based methods in the autonormal simulations, while the same free-boundary neighborhood structure is used to construct the SME-BETEL moment functions. For the misspecified working model, $\beta_h$ and $\beta_v$ are replaced by the common parameter $\alpha_h$, while $\beta_d$ is replaced by $\alpha_d$.

\subsection{Kernel Convolution Approximation for $\mathrm{Mat\acute{e}rn}$ Gaussian Process}
\label{subsec:kernel_form}

To reduce the computational cost associated with covariance matrix inversion, we use a kernel convolution approximation for the spatial Gaussian processes. Let $\delta(\boldsymbol{s})$ be a zero-mean Gaussian process with covariance function $c(r\mid\boldsymbol{\psi})$. A process convolution representation \citep{higdon2002space} writes
$$
    \delta(\boldsymbol{s}) = \int_{\mathcal{D}} K_{\boldsymbol{\psi}}(\boldsymbol{s} - \boldsymbol{u}) d W(\boldsymbol{u}),
$$
where $W$ denotes Gaussian white noise on $\mathcal{D}$ and $K_{\boldsymbol{\psi}}$ is a kernel indexed by $\boldsymbol{\psi}$. For the $\mathrm{Mat\acute{e}rn}$ covariance function in two dimensions, the corresponding kernel is
$$
    K_{\boldsymbol{\psi}}(\boldsymbol{u}) = \tau\frac{\Gamma(\nu + 1)^{1/2}\nu^{\nu / 4 + 1/4}|\boldsymbol{u}|^{\nu/2 - 1/2}}{\pi^{1/2}\Gamma(\nu/2 + 1/2)\Gamma(\nu)^{1/2}\rho^{\nu/2 + 1/2}}\mathcal{K}_{\nu/2 - 1/2}(\frac{2\nu^{1/2}|\boldsymbol{u}|}{\rho}).
$$
The convolution representation motivates a finite-dimensional approximation based on spatial knots. Let $\boldsymbol{\phi}_1,\ldots,\boldsymbol{\phi}_N$ denote a grid of knots in the spatial domain. Then, for sufficiently large $N$, we approximate
$$
    \delta(\boldsymbol{s}) \approx \sum_{j=1}^N K_{\boldsymbol{\psi}}(\boldsymbol{s} - \boldsymbol{\phi}_j) w_j,
$$
where $w_j\sim N(0,1)$ independently. This low-rank representation is used in the preferential sampling model to approximate the residual spatial processes.

\section{Metropolis--Hastings Algorithm for SME-BETEL Posterior}
\label{app:sampling_algo}

\begin{algorithm}
\caption{Metropolis--Hastings sampler for the SME--BETEL posterior}
\label{alg:mh-sme-betel}
\begin{algorithmic}
\scriptsize
\STATE \textbf{Input:} Number of iterations $S$, proposal distribution
$p_{\mathrm{prop}}(\cdot \mid \cdot)$, initial value $\boldsymbol{\theta}^{0}$
\STATE \textbf{Data:} $\mathbf{x}_{1:n}$
\FOR{$s = 0, \ldots, S-1$}
    \STATE Draw a proposal
    $\tilde{\boldsymbol{\theta}} \sim p_{\mathrm{prop}}(\cdot \mid \boldsymbol{\theta}^{s})$
    \STATE Generate $u \sim \mathrm{Unif}(0,1)$
    \IF{the origin is not contained in the interior of the convex hull of
    $\{\mathbf{g}(\mathbf{x}_i, \tilde{\boldsymbol{\theta}})\}_{i=1}^n$}
        \STATE Set $\boldsymbol{\theta}^{s+1} \gets \boldsymbol{\theta}^{s}$
    \ELSE
        \STATE Compute
        \[
        r =
        \frac{
        \pi(\tilde{\boldsymbol{\theta}})\,
        L(\tilde{\boldsymbol{\theta}})\,
        p_{\mathrm{prop}}(\boldsymbol{\theta}^{s} \mid \tilde{\boldsymbol{\theta}})
        }{
        \pi(\boldsymbol{\theta}^{s})\,
        L(\boldsymbol{\theta}^{s})\,
        p_{\mathrm{prop}}(\tilde{\boldsymbol{\theta}} \mid \boldsymbol{\theta}^{s})
        },
        \]
        where $L(\boldsymbol{\theta})$ is defined in \eqref{likelihood}.
        \IF{$u \le \min\{1,r\}$}
            \STATE Set $\boldsymbol{\theta}^{s+1} \gets \tilde{\boldsymbol{\theta}}$
        \ELSE
            \STATE Set $\boldsymbol{\theta}^{s+1} \gets \boldsymbol{\theta}^{s}$
        \ENDIF
    \ENDIF
\ENDFOR
\end{algorithmic}
\end{algorithm}

\end{document}